\documentclass[aps,prx,twocolumn,superscriptaddress,export,longbibliography]{revtex4-1}
\usepackage[colorlinks=true,linkcolor=blue,citecolor=blue,urlcolor=blue]{hyperref}
\usepackage[utf8]{inputenc}
\usepackage[german, english]{babel}
\usepackage[dvipsnames]{xcolor}
\usepackage{amsmath,amsfonts, amssymb, amsthm, dsfont}
\usepackage{bm}
\usepackage{graphicx}
\usepackage{tikz}
\usepackage{physics}
\usepackage{xspace}
\usepackage[export]{adjustbox}
\usepackage{subfigure}
\newtheorem{lemma}{Lemma}
\usepackage{microtype}

\newcommand{\qck}{$\mathrm{QC}_k$\xspace}
\newcommand{\fdqck}{$\mathrm{FDQC}_k$\xspace}
\DeclareMathOperator{\ind}{ind}

\begin{document}

\title{Non-local finite-depth circuits for constructing SPT states \\ and quantum cellular automata}

\date{\today}
\author{David T. Stephen}
\email{davidtstephen@gmail.com}
\affiliation{Department of Physics and Center for Theory of Quantum Matter, University of Colorado Boulder, Boulder, Colorado 80309 USA}
\affiliation{Department of Physics, California Institute of Technology, Pasadena, CA 91125, USA}
\author{Arpit Dua}
\email{adua@caltech.edu}
\affiliation{Department of Physics, California Institute of Technology, Pasadena, CA 91125, USA}
\author{Ali Lavasani}
\affiliation{Kavli Institute for Theoretical Physics, University of California, Santa Barbara, CA, 93106}
\author{Rahul Nandkishore}
\affiliation{Department of Physics and Center for Theory of Quantum Matter, University of Colorado Boulder, Boulder, Colorado 80309 USA}

\begin{abstract}
    Whether a given target state can be prepared by starting with a simple product state and acting with a finite-depth quantum circuit is a key question in condensed matter physics and quantum information science. It underpins classifications of topological phases, as well as the understanding of topological quantum codes, and has obvious relevance for device implementations. Traditionally, this question assumes that the quantum circuit is made up of unitary gates that are {\it geometrically local}.  
    Inspired by the advent of noisy intermediate-scale quantum devices, we reconsider this question with {\it $k$-local} gates, i.e. gates that act on no more than $k$ degrees of freedom, but are not restricted to be geometrically local. First, we construct explicit finite-depth circuits of symmetric $k$-local gates which create symmetry-protected topological (SPT) states from an initial a product state. Our construction applies both to SPT states protected by global symmetries and subsystem symmetries, but not to those with higher-form symmetries, which we conjecture remain nontrivial.
    Next, we show how to implement arbitrary translationally invariant quantum cellular automata (QCA) in any dimension using finite-depth circuits of $k$-local gates. {These results imply that the topological classifications of SPT phases and QCA both collapse to a single trivial phase in the presence of $k$-local interactions.} {We furthermore argue that SPT phases are fragile to \textit{generic} $k$-local symmetric perturbations.} We conclude by discussing the implications for other phases, such as fracton phases, and surveying future directions. 
    Our analysis opens a new experimentally motivated conceptual direction examining  the stability of phases and the feasibility of state preparation  without the assumption of geometric locality.
   \end{abstract}

 \maketitle

The exploration of {\it topological phases of matter} has been a major theme of modern condensed matter physics (for introductions, see \cite{wen,Wen2017}), with far-reaching implications for quantum information and quantum computation (for introductions, see \cite{Bombin, Pachos}). 
These phases of matter, when defined on lattices, have been classified using the complexity of preparation of the associated ground states using local quantum circuits \cite{Chen2010}; 
a state that can be accessed by starting with a product state and acting with a finite-depth quantum circuit of geometrically local gates (FDQC) can be said to be `easy' to prepare, and one that cannot be, is difficult. This characterization is intimately related to the notion of {\it topological stability}, which states that topological phases are robust to geometrically local perturbations (which should be symmetry restricted in the case of symmetry-protected topological (SPT) phases \cite{spt,Chen2013}), since topological stability implies that topologically non-trivial states cannot be connected to trivial states by a (symmetric) FDQC \cite{Bravyi2006,Chen2010}. Recently, the notion of FDQCs in relation to topological phases has been extended in various directions, such as allowing projective measurements \cite{Piroli2021,Bravyi2022,Tantivasadakarn2022a}, or by extending to linear depth quantum circuits \cite{HuangChen,Liu2022,Wang2022}, but still demanding geometric locality. The geometric locality is, of course, a very natural constraint to impose in traditional solid-state settings. However, we are witnessing the rapid development of experimental capabilities in the fields of quantum simulation and noisy intermediate-scale quantum devices (for a recent review, see \cite{Altman}), which have come to provide a new context in which to explore topological phases (see e.g. \cite{LukinTopo, TerhalTopo}). In this new setting, the geometric locality is not necessarily guaranteed and there can arise {\it $k$-local} interactions i.e., interactions that are few-body (acting on no more than $k$ degrees of freedom), but not restricted to be geometrically local (see e.g. \cite{MartinisNonlocal, MonroeNonlocal, LevNonlocal,Bourassa_2021}). Do the results on the difficulty of topological state preparation survive in this new setting, and related, does the classification of topological phases of matter and the associated notion of topological stability survive in a setting where there can be $k$-local perturbations?  

Another topological classification that may be affected by $k$-local interactions is that of locality-preserving unitary operators, also known as quantum cellular automata (QCA) \cite{Farrelly2020}.
An FDQC is an obvious example of a QCA, but there are also QCA that cannot be written as an FDQC, such as the lattice translation operator. The topological classification of phases of QCA arises from the question of whether two QCA can be smoothly connected along a locality-preserving path, with the topologically trivial phase consisting of FDQCs as they can be connected to the identity. In one and two dimensions, this classification is given by an index theory that essentially shows that every QCA can be decomposed as an FDQC and a translation \cite{Gross2012,Freedman2020}. In three dimensions, the classification is incomplete as there appear to be QCA that are neither FDQCs nor translations, but significant progress is being made \cite{Haah2022,Haah2021,Shirley2022}. The classification of QCA also plays a role in the classification of topological phases \cite{Stephen2019,Haah2021,Haah2022,Shirley2022,Fidkowski2020}, particularly Floquet phases \cite{Po2016,Fidkowski2019,Poetal}.

In this work, we examine the fate of the `topological stability' of topological phases and QCA in the presence of $k$-local perturbations by asking whether one can construct a reference state or QCA in the putatively topological phase using an FDQC made out of $k$-local (but not geometrically local) gates. 
We show that SPT phases with global or subsystem symmetries are {\it not} stable to $k$-local perturbations by explicitly constructing FDQCs of symmetric $k$-local gates that trivialize the fixed-point states by mapping them to symmetric product states. {Since all states within a given phase can be connected to the fixed-point state by a symmetric FDQC \cite{Chen2010}, our circuits for fixed-point states imply the existence of similar circuits that disentangle any state in a given SPT phase. We also argue that SPT order should be unstable to generic symmetric $k$-local noise}. Similarly, we show that every translationally-invariant QCA on a periodic lattice in any spatial dimension can be realized as an FDQC of $k$-local gates that commute with the same symmetries as the QCA. Therefore, all QCA belong to the same trivial phase when $k$-local gates are allowed. However, $k$-local perturbations do not trivialize everything - indeed it is known that topological order (such as the toric code model) is stable to $k$-local perturbations \cite{Aharonov2018}. More generally, it is known that code states of quantum error correcting codes can not be disentangled in finite-depth by $k$-local unitary gates \cite{Bravyi2006,Aharonov2018} (see Appendix~\ref{app:QECC_klocal_nontriviality} for a brief review of the proof).
We argue that fracton order and Floquet topological codes are also stable to $k$-local perturbations, and we conjecture the same is true of SPT order protected by higher-form symmetries. 

{
While our main focus is on `in principle' topological classifications in the absence of geometric locality, our work also has practical implications. In particular, we note that there has recently been a significant effort to prepare and analyze topological phases in quantum simulators and quantum computers \cite{Azses2020,Xiao2021,Semeghini2021,Smith2022,Lu2022,Mi2022,Dumitrescu2022}. In this context, our results show that certain operations which require a circuit depth that is linear in system size with local interactions can be done in finite depth using $k$-local interactions, significantly reducing the time required for implementation.}

In summary, this paper introduces a new (experimentally motivated) perspective on topological stability, showing that certain topological classifications can collapse in the presence of $k$-local interactions, and also opens up a new route for the efficient preparation of certain topological states on quantum devices. 

The rest of this paper is organized as follows. In Sec.~\ref{sec: def} we provide a precise definition of $k$-local circuits. In Sec.~\ref{sec:spt} we show that these circuits can trivialize arbitrary SPTs protected by global or subsystem symmetries, but we provide evidence that this is not possible for SPTs protected by higher-form symmetries. In Sec.~\ref{sec:qca} we show that these circuits can also trivialize arbitrary QCA. In Sec.~\ref{sec:randomcirc} we present numerical evidence that such circuits can trivialize SPT states in monitored circuits. We conclude in Sec.~\ref{sec:discussion}, where we also discuss the $k$-local non-triviality of states with true topological order (including fracton phases and Floquet topological codes), as well as discussing the implications of our results and some future directions. In the Appendix, we discuss an alternative classification based on {\it finite-time} $k$-local circuits, which are even more powerful than the finite-depth $k$-local circuits discussed in the main paper.

\section{Definition of $k$-local circuits}
\label{sec: def}
{
In this section, we lay out the definitions of the notions of locality that we will be concerned with throughout the paper. The first and most common notion of locality is geometric locality. Any lattice system has a natural notion of distance between pairs of sites, which allows us to define a geometrically local unitary gate as one which acts only on sites that can be contained within a ball of some finite radius. Here and throughout, the word finite means independent of system size, \textit{i.e.} finite even in the thermodynamic limit. Naturally, such a gate acts only on a finite number of sites in the lattice. One can construct a circuit of geometrically local gates
\begin{equation}
    U=\prod_{\ell=1}^D\left( \prod_{i} u_{\ell,i} \right)
\end{equation}
where each gate $u_{\ell,i}$ is geometrically local and the gates within a given layer $\ell$ have non-overlapping support, such that they can be applied in parallel. We call such a unitary $U$ a quantum circuit (QC). When the number of layers $D$ is finite, $U$ is called a finite-depth quantum circuit (FDQC). It is also interesting to define symmetric QCs, which are circuits in which each gate individually commutes with some symmetry operator.

We can extend the notion of the geometric locality to $k$-locality. A gate is called $k$-local if it has support on at most $k$ sites. Thus a $k$-local gate shares the few-body property of a geometrically local gate while ignoring the relative position and distance between spins. We define a $k$-local quantum circuit, \qck, to be composed of $k$-local gates $u_{\ell,i}$ where $k$ is finite, such that the gates are applied in layers and in each layer they have non-overlapping support. When the number of layers is finite, we call the circuit a finite-depth \qck (\fdqck).

Clearly, every QC is also a \qck, but in general, a \qck is more powerful in the sense that writing a many-body unitary operator as a \qck can sometimes be accomplished with lower depth than is required to write it as a QC. As a simple example, consider the unitary operator which generates the $N$-qubit Greenberger–Horne–Zeilinger (GHZ) state from a product state. This can only be implemented with a linear-depth QC due to the presence of long-range correlations in the GHZ state \cite{Aharonov2018}. However, it can be done with a \qck having a depth that is logarithmic in $N$, see Refs.~\cite{Cruz2019,Liao2021} for example. Similarly, the unitary which generates the two-dimensional toric code ground state, an example of topological order, requires a linear-depth QC \cite{Bravyi2006} but can be done in a log-depth \qck \cite{konig2009exact,Liao2021}. Notably, neither of these unitaries can be written as an \fdqck \cite{Aharonov2018}, so we say that the GHZ and toric code states remain non-trivial in the $k$-local setting. In contrast, the unitaries discussed in this paper will primarily be circuits that require linear depth using a (symmetric) QC, but can be performed in finite-depth using a (symmetric) \qck, so we say they become trivial in the $k$-local setting.
}

\section{$k$-local instability of SPT\lowercase{s}} \label{sec:spt}

In this section, we ask whether SPT states can be prepared by symmetric FDQC$_k$. In Secs.~\ref{sec:1D} and \ref{sec:2D} we consider global on-site symmetries. While it is known that a linear-depth symmetric QC is needed to create an SPT ground state \cite{HuangChen} from a symmetric product state, we will show that a symmetric \fdqck is sufficient. We first give an intuitive physical argument as to why this is the case by studying the boundaries of SPT phases, and we then construct an explicit finite-depth $k$-local symmetric circuit that disentangles fixed-point SPT states. In Sec.~\ref{sec:SSPT}, we show that SPT phases protected by subsystem symmetries (SSPT phases) can also become trivial in the $k$-local scenario. Conversely, we argue in Sec.~\ref{sec:higherformSPT} that SPT phases protected by higher-form symmetries remain non-trivial even in the $k$-local scenario. In Appendix \ref{app:finitetime}, we give alternative finite-time $k$-local constructions of SPT states which have the advantage that all interactions are local except for one special qubit which can interact with all others.

\subsection{1D SPT phases with global symmetries} \label{sec:1D}

Let us first consider 1D SPT phases. The characteristic feature of 1D SPT order is the existence of zero-energy edge modes that are protected by the bulk symmetry. On periodic boundaries, the ground state is unique. But when an edge is introduced, {degenerate ground states which differ only in a region exponentially close to the boundary will appear}. This boundary degeneracy is robust in the sense that no local, symmetric perturbation can split the degeneracy. However, the degeneracy is not robust to $k$-local interactions, as a symmetric interaction can be used to couple the two edges in such a way as to split the degeneracy. This suggests that 1D SPT orders are not robust to $k$-local symmetric interactions. However, there may still be some non-trivial bulk properties that cannot be removed by $k$-local interactions. We show that this is not the case by explicitly constructing a symmetric \fdqck which maps generic SPT fixed-point states to product states. 

We can understand why such a circuit should exist using a simple folding argument. Namely, consider a 1D SPT ordered state $|\psi\rangle$ on a ring. Suppose that we ``fold'' the ring, bringing the two opposite sides close to each other. In the bulk of the folded system, it looks like we have stacked the state with its spatially reversed self, see Fig.~\ref{fig:1d_disentangler}. However, it is well known that SPT phases are invertible, meaning that there exists a second state $|\psi^{-1}\rangle$ such that the joint system $|\psi\rangle\otimes |\psi^{-1}\rangle$ is in a trivial SPT phase. For this statement to make sense, it is important to specify how the symmetry acts on the joint system. If $|\psi\rangle$ has symmetry $G$ with an on-site representation $U(g)$, then the symmetry acts on the joint system ``diagonally'', \textit{i.e.} with representation $U(g)\otimes U(g)$ of $G$. If we instead considered the symmetry group $G\times G$ represented by $U(g)\otimes U(h)$ with $g$ not necessarily equal to $h$, then the stacked system is still in a nontrivial phase with respect to this larger symmetry. In the folded system, the global symmetry indeed acts in the same way across the whole ring, so the total symmetry group is still only $G$.

It turns out that the spatial inverse of the SPT fixed-point states is in fact the inverse state in the above sense \cite{Chen2013}, so the bulk of the folded system looks like a trivial SPT phase. Then, there must be a symmetric finite-depth circuit that disentangles the bulk to a product state. Looking at this picture without folding, the disentangling circuit couples the two distant edges of the ring, therefore requiring long-range $k$-local gates.

\begin{figure}
    \centering
    \includegraphics[width=\linewidth]{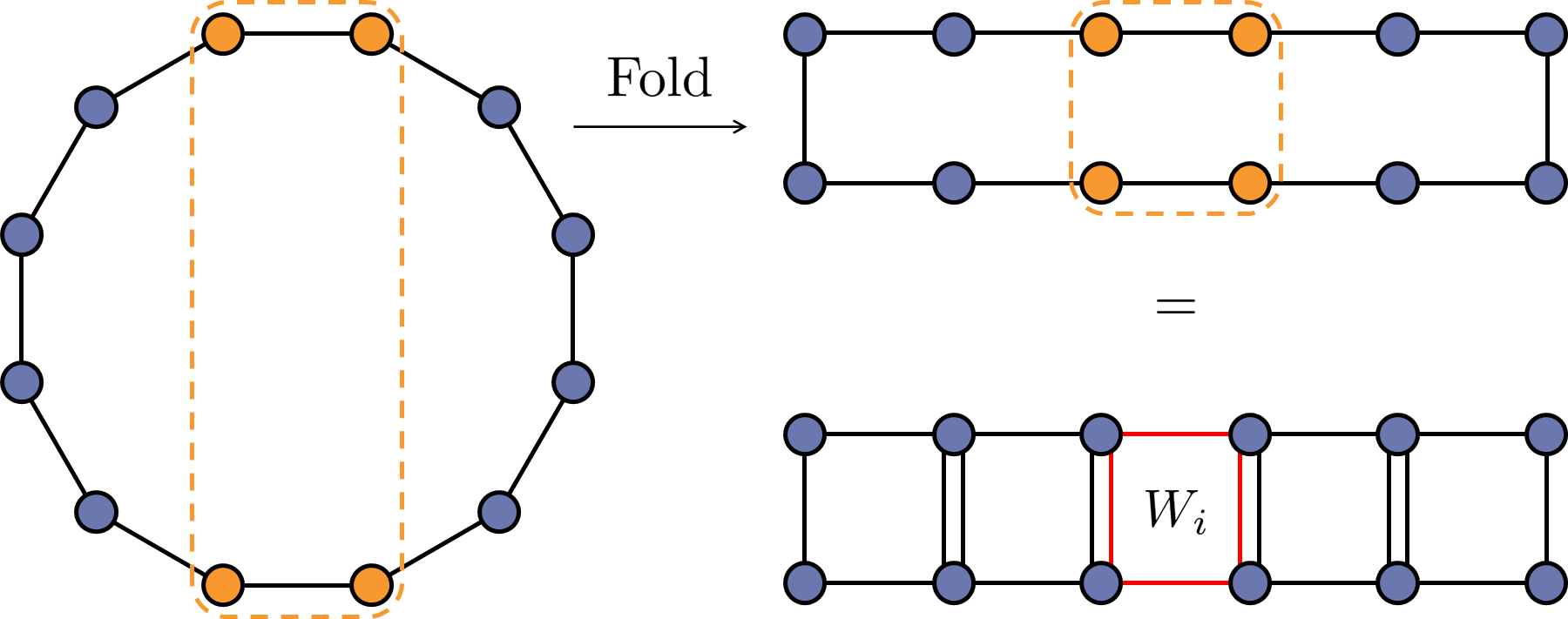}
    \caption{Disentangling 1D SPT phases with $k$-local symmetric gates. After folding a 1D chain with periodic boundary conditions in half, the resulting system looks like a stack of the chain and its spatial inverse in the bulk. This stack can be disentangled with symmetric gates acting on the highlighted qubits, which correspond to $k$-local gates in the original system. For fixed-point SPT states, these gates are the $W_i$ defined in Eq.~\ref{eq:wi_cluster}, one such gate is highlighted in red where each line is a $CZ$ gate. Note that the adjacent vertical gates cancel each other pairwise in the bulk.}
    \label{fig:1d_disentangler}
\end{figure}

Let us explicitly construct such a circuit. We give only a single example here, as the general case is covered by the construction in Sec.~\ref{sec:qca}.
The example we consider is the 1D cluster state \cite{Briegel2001}, which can be created from a product state using a finite circuit in the following way,
\begin{equation} \label{eq:1d_clus}
    |\psi_C\rangle = \left(\prod_{i=1}^N CZ_{i,i+1}\right) |+\rangle^{\otimes N}
\end{equation}
where $|+\rangle = \frac{1}{\sqrt{2}}(|0\rangle + |1\rangle)$ and $CZ = I - 2|11\rangle\langle 11|$ is the controlled-$Z$ gate. Assuming $N$ is even, $|\psi_C\rangle$ has a $\mathbb{Z}_2\times\mathbb{Z}_2$ symmetry generated by $X_{\mathrm{odd}} = \prod_{i=1}^{N/2} X_{2i+1}$ and $X_{\mathrm{even}} = \prod_{i=1}^{N/2} X_{2i}$. 

Importantly, while the circuit of controlled-$Z$ gates commutes with this symmetry as a whole, the individual gates do not. Since the 1D cluster state has non-trivial SPT order \cite{Son2012}, there does not exist a symmetric FDQC which maps it to a product state \cite{Chen2011,HuangChen}. However, it is possible using a symmetric \fdqck. Consider the gates,
\begin{equation} \label{eq:wi_cluster}
    W_i = CZ_{i,i+1} CZ_{i+1,N-i} CZ_{N-i,N-i+1} CZ_{N-i+1,i},
\end{equation}
which are depicted in Fig.~\ref{fig:1d_disentangler}. It is straightforward to check that $W_i$ commutes with the $\mathbb{Z}_2\times\mathbb{Z}_2$ symmetry. Furthermore, we have,
\begin{equation} \label{eq:1dspt_disent}
    \left( \prod_{i=1}^{N/2-1} W_i \right)|\psi_C\rangle = |+\rangle^{\otimes N},
\end{equation}
as shown in Fig.~\ref{fig:1d_disentangler}. Since the $W_i$ all commute with each other, and since the support of each $W_i$ overlaps with the support of finitely many others, they can be applied in a finite number of layers such that this is a symmetric \fdqck which trivializes the cluster state.

It is instructive to note that $W_i$ is the operator that creates a small 4-site cluster state on a ring, which explains why it is symmetric. The $k$-local disentangling circuit can therefore be interpreted as a ``bubbling'' procedure in which one decomposes the $N$-site cluster state into a number of small 4-site cluster states and then disentangles each 4-site cluster state with a symmetric $k$-local gate. 

\subsection{2D and higher-dimensional SPT phases} \label{sec:2D}

\begin{figure}
    \centering
    \includegraphics[width=0.8\linewidth]{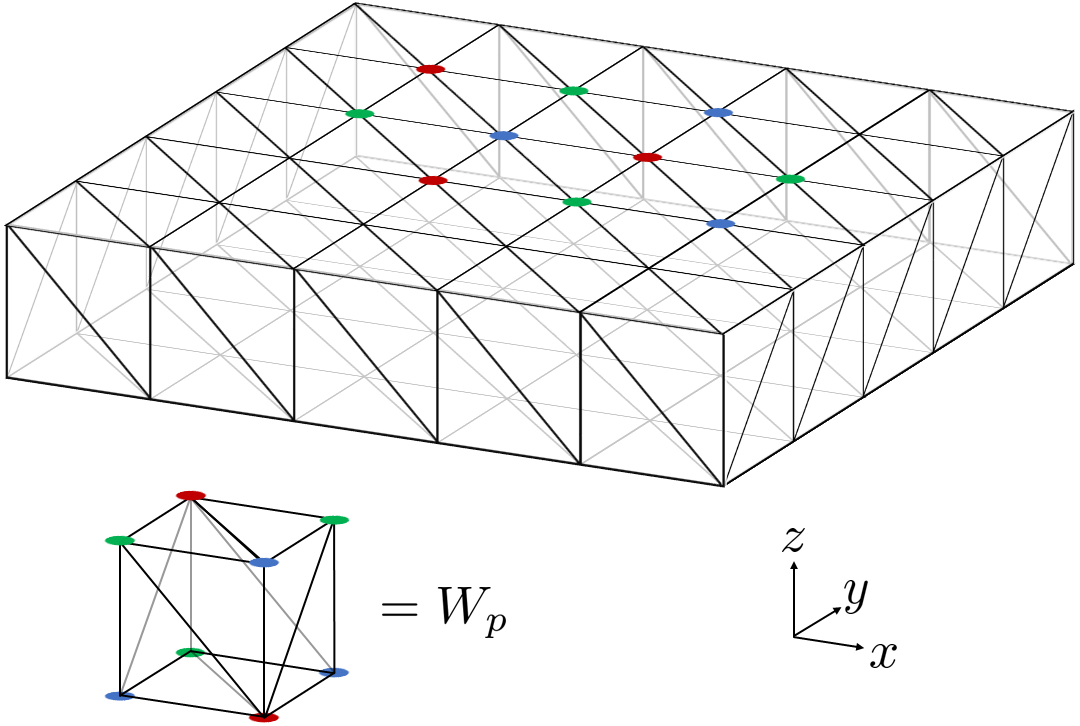}
    \caption{Disentangling 2D SPT phases with $k$-local symmetric gates. We show a folded system with spherical boundary conditions, which leads to a thin rectangular prism. A torus is obtained by identifying the left and right faces. A subset of the lattice coloring is shown. The rectangular prism geometry can be decomposed into smaller rectangular prisms. The symmetric gate $W_p$ is defined by acting with $CCZ$ on all 12 of the triangular faces of this prism. When applying $W_p$ to every prism, all bulk gates that act on faces parallel to the $z$-axis cancel between neighboring prisms, leaving only gates on the surface of the large rectangular prism.}
    \label{fig:2d_disentangler}
\end{figure}

The folding argument from the previous section carries over equally well to SPT phases with global symmetries in 2D and higher. We give only the 2D argument explicitly, as the generalization to higher dimensions is straightforward. We again give a single example as the general case (in all dimensions) is covered by Sec.~\ref{sec:qca}.

As in 1D, we can predict that 2D SPT phases are trivial with $k$-local interactions by considering the boundary. Consider the example of a topological insulator, which has a gapless helical edge on the boundary of a disk. With local symmetric interactions, the edge cannot be gapped out without breaking the symmetry (although see \cite{ChouTI, KimchiTI}). However two opposite points on the boundary have helical currents moving in opposite directions, which could backscatter off each other if coupled by a $k$-local perturbation. Therefore, it is possible to introduce a $k$-local symmetric term that couples opposite points and gaps out the edge without breaking symmetry. 

As before, we will confirm this intuition by constructing explicit symmetric \fdqck disentanglers. As our example, we choose the 2D hypergraph state, first defined in Ref.~\cite{Yoshida2016}. The state is defined on a triangular lattice with one qubit per site. We choose boundary conditions of a sphere for our demonstration, but a torus would work equally well. The state is defined as follows,
\begin{equation} \label{eq:hypergraph}
    |\psi_H\rangle = \left(\prod_{\triangle} CCZ_\triangle\right) |+\rangle^{\otimes N}
\end{equation}
where $CCZ$ acts on the three qubits around a triangle as $CCZ = I-2|111\rangle\langle 111|$. This state has SPT order protected by a $\mathbb{Z}_2\times\mathbb{Z}_2\times\mathbb{Z}_2$ symmetry \cite{Yoshida2016}. This symmetry relies on the fact that the lattice is 3-colorable, meaning that each site can be assigned a color (red, blue, or green) such that neighboring sites have different colors, see Fig.~\ref{fig:2d_disentangler}. The symmetry is then generated by the operators $X_R$, $X_B$, and $X_G$, which are tensor products of $X$ on every red, blue, and green site, respectively. We remark that this state is closely related to the Levin-Gu state \cite{Levin2012} which is an example of 2D SPT order with $\mathbb{Z}_2$ symmetry \footnote{The Levin-Gu state can be obtained by acting on $|\psi_H\rangle$ with $CZ$ on every neighboring pair of spins and $Z$ on every spin \cite{Wei2018a}.}. As before, the circuit of $CCZ$'s is symmetric as a whole on closed boundary conditions, but the individual gates are not symmetric.

By folding the 2D system, we get a state defined on the surface of a thin rectangular prism. This rectangular prism can be decomposed into a number of small rectangular prisms with triangular faces, see Fig.~\ref{fig:2d_disentangler}. We label these prisms by $p\in P$. For each such prism $p$, we define a gate,
\begin{equation}
    W_p = \prod_{\triangle\in p} CCZ_\triangle,
\end{equation}
which is depicted in Fig.~\ref{fig:2d_disentangler}. As in the 1D case, this operator can be interpreted as creating a small instance of $|\psi_H\rangle$ on a triangulation of a sphere, and it, therefore, respects the $\mathbb{Z}_2\times\mathbb{Z}_2\times\mathbb{Z}_2$ symmetry. Applying this operator to every prism, we have,
\begin{equation}
    \left(\prod_{p\in P} W_p\right) |\psi_H\rangle = |+\rangle^{\otimes N},
\end{equation}
so this is a symmetric \fdqck disentangler.

\subsection{SPT phases with subsystem symmetries}
\label{sec:SSPT}

\begin{figure}
    \centering
    \includegraphics[width=0.8\linewidth]{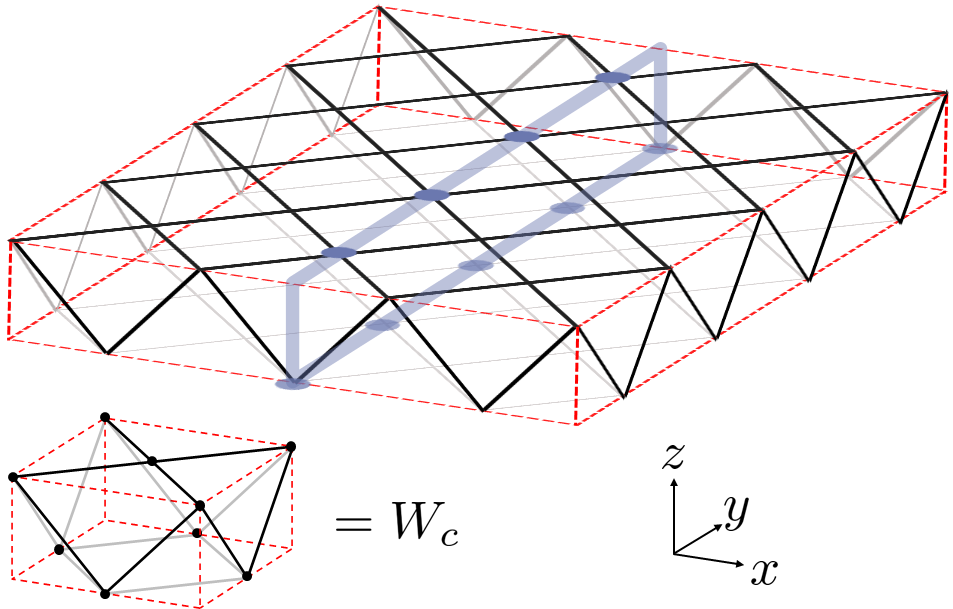}
    \caption{Disentangling 2D SSPT phases with $k$-local symmetric gates. The dashed lines are shown as a guide for the eye. A torus is obtained by identifying the left and right opposite faces. The gate $W_c$ consists of $CZ$s on every solid edge, with qubits on located at the small dots (not drawn on the upper figure). The thick solid loop and dots indicate the support of one line-like subsystem symmetry generator.
    }
    \label{fig:sspt_disentangler}
\end{figure}

We now turn to SPT phases protected by subsystem symmetries. These are symmetries that act on rigid, lower-dimensional submanifolds of the entire system, such as straight lines across a 2D lattice \cite{Raussendorf2019,You2018,Devakul2019,Devakul2018class,You2020}. These are similar to 1D SPTs, in that they are characterized by an extensive degeneracy on the edge \cite{You2018,Devakul2019}. Indeed, if a 2D SSPT phase on a cylinder is treated like a quasi-1D system along the cylinder's length, then it behaves like a 1D SPT phase with a sub-extensive number of symmetry generators \cite{Stephen2019}.

The prototypical SSPT phase is represented by the 2D cluster state \cite{Briegel2001}. This state consists of qubits on a 2D square lattice and is defined as follows,
\begin{equation}
    |\psi_{2DC}\rangle = \left(\prod_{\langle ij\rangle} CZ_{i,j}\right)|+\rangle^{\otimes N},
\end{equation}
where the product is over all nearest neighbors in the square lattice. The symmetries of this model form rigid diagonal lines spanning the square lattice, defined as,
\begin{equation}
    U_{c,\pm} = \prod_x X_{(x,c\pm x)}
\end{equation}
where $i=(x,y)$ is a coordinate on the 2D square lattice. Similar to the cases of global symmetry, one can use the folding trick to create $|\psi_{2DC}\rangle$ on a closed 2D manifold using a symmetric \fdqck consisting of the symmetric $k$-local gates $W_c$ defined in Fig.~\ref{fig:sspt_disentangler}. 

We remark that in order for a diagonal line of symmetry to act the same on the top and bottom layers of the fold as in Fig.~\ref{fig:sspt_disentangler}---which is necessary for the folding argument to work---we needed to put $|\psi_{2DC}\rangle$ on a $45^\circ$-rotated square lattice. We note that this results in some triangular faces appearing on spherical boundaries (which are absent on the torus). For general subsystem symmetries, which can have more complex geometries such as fractal geometries \cite{Devakul2018class}, we expect that a symmetric \fdqck will only be able to create SSPT states on spatial manifolds with nice enough geometry such that there are distant regions in space where the symmetry mirrors itself.

\subsection{SPT phases with higher-form symmetries}
\label{sec:higherformSPT}

In this section, we argue that in contrast to global and subsystem symmetries, phases protected by higher-form symmetries are robust to $k$-local symmetric interactions. In general, a $q$-form symmetry is one that acts on closed codimension-$q$ submanifolds of space. In contrast to subsystem symmetries, these submanifolds are not rigid and can be deformed freely. These symmetries are local in the sense that they may have support on only a finite number of sites. For example, a 1-form symmetry in 3D is generated by objects acting on any closed 2D surface embedded in the 3D space. {A simple example of a state with SPT order under higher-form symmetries is the 3D cluster state defined in Ref.~\cite{Raussendorf2005a} (see also Ref.~\cite{Roberts2020} for a more detailed discussion of the SPT order) which is also an example of the general Walker-Wang construction \cite{Walker2012}.}

We first observe that the folding argument used for global symmetries does not carry over to the case of higher-form symmetries. This is because folding, in this case, does not correspond to stacking in the usual sense. Recall that stacking requires the symmetry to act in the same way on the two layers. When we fold a system with global symmetry, acting with the symmetry on the whole system automatically has identical action on the two layers of the fold. However, this is not the case when we fold a system with higher-form symmetry since the symmetry generators, being local, can act independently on either layer. Because of this, the bulk of the folded system resembles two stacked SPT states, each with its own independent symmetry, and such a system has a non-trivial SPT order. 

{
Second, we observe that any $k$-local interaction that commutes with a higher-form symmetry is also locally symmetric, meaning the gates can be decomposed into a sum of tensor products of geometrically local symmetric unitaries. This is due to the simple fact that the higher-form symmetry group itself contains operators that act non-trivially only in local regions of space. This is much different from the case of global symmetries, where a symmetric $k$-local interaction can violate symmetry locally while still preserving it globally. Because of this, the class of $k$-local gates that we can employ is severely restricted compared to the case of global symmetry. Indeed, the \fdqck{}'s used to disentangle the SPT states above only commute with the symmetry globally as they contain operators such as a long-range $ZZ$ pair, which can transfer symmetry charge over long distances. In Sec.~\ref{sec:generic}, we argue that the fact that symmetric $k$-local interactions can violate the symmetry in local regions of space is crucial to their ability to destroy SPT order. Accordingly, we expect that $k$-local gates that preserve the symmetry locally, which is always the case for higher-form symmetric gates, are insufficient to destroy SPT order.
}

Finally, examining the boundary of a higher-form SPT also suggests that it may be $k$-local non-trivial. Higher-form SPTs in 3D, for example, can support topologically ordered boundary theories such as a 2D toric code appearing on the boundary of the 3D cluster state \cite{Walker2012,Yoshida2016,Roberts2020}. However, unlike in the boundaries of systems with global symmetry discussed above, $k$-local interactions cannot trivialize this boundary theory because the 2D toric code, and any number of stacks of it, are $k$-local non-trivial \cite{Aharonov2018}.  

\section{Unified construction of symmetric $k$-local circuits for QCA} \label{sec:qca}

\begin{figure}
    \centering
    \includegraphics[width=\linewidth]{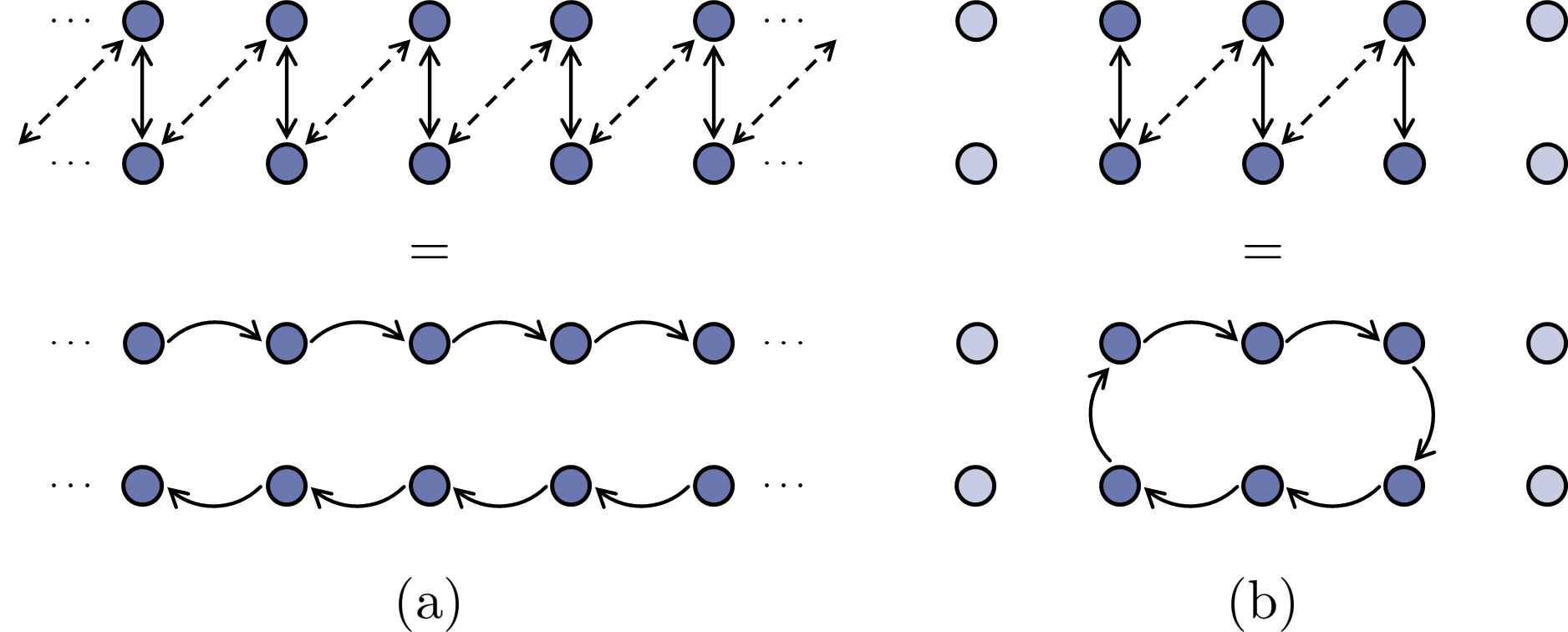}
    \caption{(a) A pair of counter-propagating 1D QCA implementing the shift operation, where information flows in the directions of the arrows, can be realized by a depth-2 $k$-local circuit of SWAP gates, indicated by double-headed arrows, where all solid arrows first act in parallel, then all dotted arrows act in parallel. (b) Truncating this circuit stitches the top and bottom lines of qubits into a single ring acted on by the shift QCA.}
    \label{fig:shift_qca}
\end{figure}

In this section, we show that locality-preserving unitary operators, also known as quantum cellular automata (QCA) \cite{Farrelly2020}, can be implemented as \fdqck{}'s. We give a generic construction of these \fdqck{}'s that works in any spatial dimension $d$. We only deal with translationally-invariant QCA, but this condition can likely be relaxed to a certain extent. We also assume the QCA acts on a lattice with mirror symmetry in all $d$ directions. 

An FDQC is a trivial example of a QCA. However, there are also QCA that cannot be expressed as an FDQC, such as the 1D shift QCA $Q_{S}$ which acts on any operator $O_i$ supported on site $i$ by translating it by one site, $Q_S O_iQ_S^{-1} = O_{i+1}$. When only local gates are employed, implementing this operation requires QC whose depth grows linearly with system size. In 1D and 2D, it has been shown that all QCA are composed of shifts and FDQCs \cite{Gross2012,Freedman2020}. In 3D, however, there are believed to be \textit{non-trivial} QCA which are neither shifts nor circuits \cite{Haah2022,Haah2021,Shirley2022}. {Our result constructs \fdqck{}'s for all QCA, which, to the best of our knowledge, provides the first circuit representation of nontrivial 3D QCA. The existence of these \fdqck{}'s implies that all QCA become trivial in the $k$-local scenario, such that the topological classification of QCA collapses to a single trivial phase when $k$-local gates are allowed.} Additionally, if the QCA commutes with some global symmetry, it may not be possible to write it as a \textit{symmetric} FDQC, even if it is an FDQC. Nonetheless, the \fdqck we construct consists of symmetric gates (in the case of global symmetries), so this result gives an alternative demonstration of the triviality of SPT order by applying our construction to the QCA which generate the SPT states, such as those in Eq.~\ref{eq:1d_clus} and Eq.~\ref{eq:hypergraph}.

The main ideas which lead to this result are the following. Given a QCA $Q$ acting in $d$ spatial dimensions, we start with the known fact that $Q\otimes Q^{-1}$ acting on two copies of a system can be realized by an FDQC \cite{Arrighi2011,Farrelly2020}. By truncating this FDQC to a finite region of space, we obtain a $k$-local circuit acting on a single conjoined system. On one half of this system, the circuit acts like $Q$, while on the other half, it acts like $Q^{-1}$. We then show that $Q^{-1}$ is equivalent to the spatial inversion of $Q$ in one direction, which we denote as $\bar{Q}$, up to a circuit of $(d-1)$-dimensional QCA that can be applied in two parallel layers. The action of $Q$ on one half and $\bar{Q}$ on the other is nothing more than $Q$ applied to a periodic system. An inductive argument then concludes that $Q$ can be implemented as a \fdqck. This idea is demonstrated in Fig.~\ref{fig:shift_qca} for the simple example of the shift QCA. In this case, the second part of the argument is not needed as $\bar{Q}$ is already equal to $Q^{-1}$. {We work out three explicit examples of the general construction, including the shift QCA, in Appendix \ref{sec:qca_ex}.}

\subsection{Construction from 1D Margolus representation}

Consider a QCA $Q$ which acts on a Hilbert space $\mathcal{H}=(\mathbb{C}^d)^{\otimes N}$. Now construct a doubled Hilbert space $\mathcal{H}_A\otimes \mathcal{H}_B$ made of two copies of $\mathcal{H}$. For every degree of freedom $i$ in $\mathcal{H}$, we have two degrees of freedom $[i]_A$ and $[i]_B$ in $\mathcal{H}_A\otimes \mathcal{H}_B$. Let $Q_{A}$ ($Q_B$) denote $Q$ acting on $\mathcal{H}_{A}$ ($\mathcal{H}_{B}$) and write $V=Q_B^{-1}Q_A$. Now write $\mathcal{S}_{AB}=\prod_i S_i$ where $S_i$ is the SWAP operation that exchanges sites $[i]_A$ and $[i]_B$. For now, we imagine that $A$ and $B$ are geometrically close such that $S_i$ is a local operator. Then we have \cite{Arrighi2011,Farrelly2020},
\begin{align}
    V &= Q_B^{-1}Q_A \nonumber \\
    &=  \mathcal{S}_{AB} Q^{-1}_A \mathcal{S}_{AB} Q_A \nonumber \\
    &= \left(\prod_i S_i\right) \left(\prod_i V_i\right),
\end{align}
where,
\begin{equation}
    V_i = Q_A^{-1} S_i Q_A.
\end{equation}
Since $S_i$ is a local operator, and $Q_A$ is locality-preserving, $V_i$ is a local operator. Since the $V_i$ all commute with each other for all $i$, the above formula can be parallelized into a finite-depth circuit realizing $V$. 

Now we imagine truncating the above circuit as follows,
\begin{equation} \label{eq:vr}
    V_R =\left(\prod_{i\in R} S_i\right) \left(\prod_{i\in R} V_i\right)
\end{equation}
where $R$ is some finite connected subset of sites. Far outside of the region $R$, $V_R$ will act as the identity. Deep inside the region $R$, $V_R$ will act as $Q_AQ^{-1}_B$. To understand what happens near the boundaries of $R$, we will employ the so-called Margolus representation of a QCA. We will first describe this for 1D systems and then show how to extend to higher dimensional systems in the next section.

Take the physical space to be a 1D chain with sites indexed by a single integer $i$. By blocking a finite number of sites, \textit{i.e.} enlarging the unit cell, it is always possible to make $Q$ have unit range, meaning that if $O_i$ is an operator supported on a site $i$, then $QO_iQ^{-1}$ is supported at most on sites $i-1,i,i+1$. Then, according to the results of Ref.~\cite{Schumacher2004}, the QCA can be written in the following Margolus representation,
\begin{equation} \label{eq:standard}
    Q = \left( \prod_{i} v_{2i-1,2i} \right)\left( \prod_{i} u_{2i,2i+1} \right)
\end{equation}
where $u$ is a unitary operator
mapping from the $d^2$-dimensional Hilbert space $\mathbb{C}^d\otimes \mathbb{C}^d$ to the $\ell r$-dimensional Hilbert space $\mathbb{C}^\ell\otimes\mathbb{C}^r$ where $d$ is the dimension of a unit cell and $\ell r=d^2$. Similarly, $v$ is a unitary operator mapping from $\mathbb{C}^r\otimes \mathbb{C}^\ell$ to $\mathbb{C}^d\otimes \mathbb{C}^d$.
Graphically, we can represent the right-hand side of this equation as,
\begin{equation} \label{eq:standard_form}
    \includegraphics[scale=0.44]{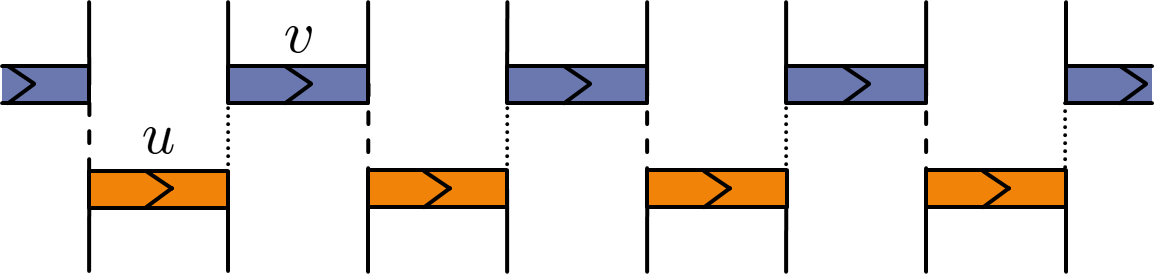}
\end{equation}
where the solid, dashed, and dotted lines have dimensions $d$, $\ell$, and $r$, respectively. The arrows indicate the orientation of the unitaries $u$ and $v$  and are expressed formally by the spatial ordering of the Hilbert space labels on the left and the right. This will become important when we invert this orientation, as we define shortly. The standard form encompasses all 1D QCA, even the shift QCA whose standard form is described in Appendix \ref{sec:qca_ex_shift}. 

{We emphasize that the $\ell$ and $r$-dimensional Hilbert spaces are not physical, and are rather a technical tool used to write the Margolus representation. Accordingly, the operators $u$ and $v$ are not proper unitary gates, since their input and output Hilbert spaces are not equivalent. Rather, they are used as building blocks to define proper unitary gates such as $Q$ itself.}

\begin{figure*}
    \centering
    \includegraphics[width=\linewidth]{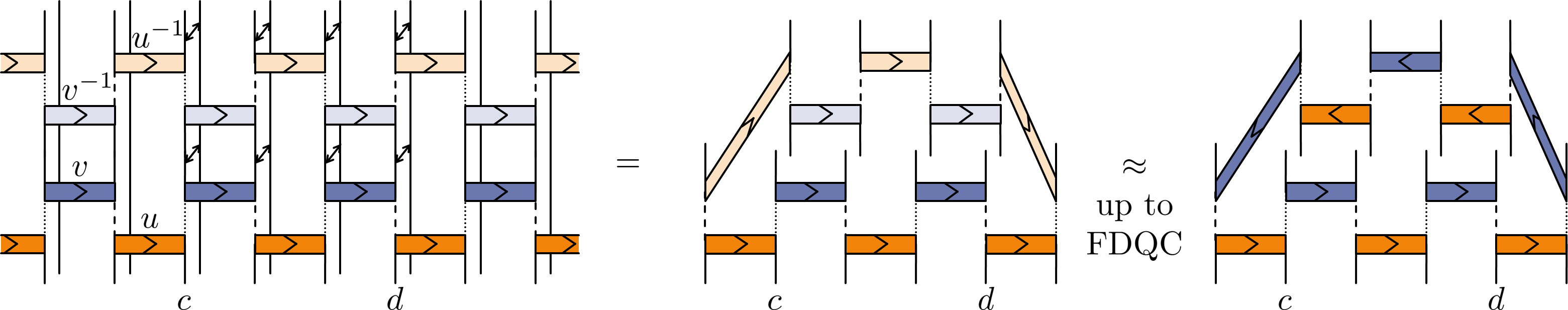}
    \caption{The truncated circuit $V_R$ where the 1D QCA $Q$ is expressed in standard form. The downwards- (upwards-) shifted vertical lines correspond to degrees of freedom in the $A$ ($B$) subsystem. The lighter-colored rectangles represent the inverses $u^{-1}$ and $v^{-1}$ as indicated in the first expression. An arrow between two lines indicates the SWAP operation $S_i$ between degrees of freedom in the two subsystems. In the second expression, most gates have canceled pairwise, leaving only those shown. {The degrees of freedom acted on by the remaining gates form the ring $\mathcal{R}$.} For visual clarity, the vertical lines have been shortened. The third expression is equivalent to the second expression up to the symmetric finite-depth quantum circuits (FDQCs) $W_1$ and $W_2$; see Eq.~\ref{eq:W1W2}. Note that the strict equivalence up to symmetric FDQCs is in one spatial dimension, $d=1$. In higher dimensions ($d>1$), we use the compactified picture such that the equivalence is up to a circuit that is finite-depth along the non-compactified direction and in general, a $(d-1)$-dimensional QCA along the $d-1$ compactified directions; see main text.}
    \label{fig:vr}
\end{figure*}

We now insert the standard form of $Q$ into the definition of $V_R$ where we take $R$ to be a segment of the 1D line $R=[c,d]$. Without loss of generality, suppose that $c$ is odd and $d$ is even such that the length of $R$ is even. We then have, 
\begin{align}
    V_R =& \left(\prod_{i=c}^d S_i\right) \left(\prod_{i=c}^d V_i \right) \nonumber \\
    =& \left(\prod_{i=c}^d S_i \right) Q_A^{-1} \left( \prod_{i=c}^d S_i\right) Q_A \nonumber \\
    =&\left(u^{-1}_{[c-1]_A,[c]_B}u^{-1}_{[d]_B,[d+1]_A}\right)
    \nonumber \\
    & \left(\prod_{i=\frac{c+1}{2}}^{d/2-1} u^{-1}_{[2i]_B,[2i+1]_B}\right) \left(\prod_{i=\frac{c+1}{2}}^{d/2} v_{[2i-1]_A,[2i]_A}\right)
    \nonumber \\
    &\left(\prod_{i=\frac{c+1}{2}}^{d/2} v^{-1}_{[2i-1]_B,[2i]_B}\right)
    \left(\prod_{i=\frac{c-1}{2}}^{d/2} u_{[2i]_A,[2i+1]_A}\right)
\end{align}
The second and third equalities above are depicted on the left and middle of Fig.~\ref{fig:vr}, respectively. As is clear from Fig.~\ref{fig:vr}, the $u$, $u^{-1}$, $v$, $v^{-1}$ operators can be parallelized into two layers of disjoint operators. Namely, the $u$ and $v^-1$ operators can act in parallel in the first layer, and the $v$ and $u^{-1}$ operators act in the second layer.
Furthermore, we see that the degrees of freedom within the range $[c-1,d+1]$ have been stitched into a single periodic 1D system, a ring of length $L=2(d-c+1)+2$ with spins in the $A$ ($B$) subsystem forming the ``front'' (``back'') of the ring that we denote as $\mathcal{R}$. Whenever we use the notation $\mathcal{R}$, we imagine it as representing a finite periodic array of sites that can be ordered counter-clockwise as $[c]_A,\dots,[d+1]_A,[d]_B,[d-1]_B,\dots,[c]_B,[c-1]_A$. 

The operator $V_R$ is not exactly $Q$ acting on the finite periodic system defined by $\mathcal{R}$, which is depicted in the third expression in Fig.~\ref{fig:vr}. Instead, $V_R$ realizes $Q$ on the front of $\mathcal{R}$ and $\bar{Q}^{-1}$ on the back of $\mathcal{R}$, with the two operators being blended near the edges.
To fix this, we show that $V_R$ is equivalent to $Q$ up to composition with an FDQC. {Given any 1D QCA $Q$}, we show that its inverse $Q^{-1}$ is related to its orientation-reversed self $\bar{Q}$ by an FDQC. This is intuitively clear for 1D QCA such as the shift QCA: reversing the direction of the shift is the same as inverting it. Generally, $\bar{Q}$ will not equal $Q^{-1}$, but they will differ only by an FDQC in 1D \footnote{Indeed, if we let $\ind(Q)$ be the rational-valued index of $Q$ as defined in Ref.~\cite{Gross2012}, then it is straightforward to show that $\ind(Q^{-1})=\ind(Q)^{-1}=\ind(\bar{Q})$, which implies that $Q^{-1}$ and $\bar{Q}$ differ by an FDQC in 1D \cite{Gross2012}.}. 

To explicitly construct the circuit which maps $\bar{Q}$ to $Q^{-1}$, we introduce the unitary $w = \bar{v}u$ where $\bar{v}$ is the spatial reversal (opposite orientation) of $v$ obtained by exchanging the left and right input and output Hilbert spaces. {That is, $v_{i,j}=\bar{v}_{j,i}$.} Graphically,
\begin{equation}
    \includegraphics[scale=0.44]{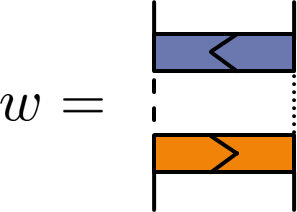}
\end{equation}
Note that $w$, unlike $u$ and $v$, is a proper unitary gate that maps $\mathbb{C}^d\otimes \mathbb{C}^d$ to itself. Observe that we have $wu^{-1}=\bar{v}$ and $v^{-1}\bar{w}=\bar{u}$ where $\bar{w}$ is again the spatial reversal of $w$. Then we have,
\begin{equation}
   \left( \prod_{i} w_{2i,2i+1} \right) Q^{-1}\left( \prod_{i} \bar{w}_{2i-1,2i} \right)=\bar{Q}.
\end{equation}
Since the $w$ ($\bar{w}$) gates on either side of $Q^{-1}$ are non-overlapping, they can be applied in parallel. Therefore,  $Q^{-1}$ and $\bar{Q}$ are related by composition with FDQCs. Using this, we define the circuits acting on the ``back'' of $\mathcal{R}$,
\begin{align}
\label{eq:W1W2}
    W_1&=\left(w_{[c-1]_A,[c]_B}w_{[d]_B,[d+1]_A}\right) \left(\prod_{i=\frac{c+1}{2}}^{d/2-1} w_{[2i]_B,[2i+1]_B}\right) \nonumber \\
    W_2&=\left(\prod_{i=\frac{c+1}{2}}^{d/2} \bar{w}_{[2i-1]_B,[2i]_B}\right),
\end{align}
such that,
\begin{align} \label{eq:qcafdqck}
    W_1V_R W_2  =&\left(\bar{v}_{[c-1]_A,[c]_B}\bar{v}_{[d]_B,[d+1]_A}\right)
    \nonumber \\
    & \left(\prod_{i=\frac{c+1}{2}}^{d/2-1} \bar{v}_{[2i]_B,[2i+1]_B}\right) \left(\prod_{i=\frac{c+1}{2}}^{d/2} v_{[2i-1]_A,[2i]_A}\right)
    \nonumber \\
    &\left(\prod_{i=\frac{c+1}{2}}^{d/2} \bar{u}_{[2i-1]_B,[2i]_B}\right)
    \left(\prod_{i=\frac{c-1}{2}}^{d/2} u_{[2i]_A,[2i+1]_A}\right)
\end{align}
The operator $W_1V_RW_2$ is 
shown in the third expression in
Fig.~\ref{fig:vr}, from which it is clear that $W_1V_RW_2$ is nothing but $Q$ applied to $\mathcal{R}$. According to the 1D ring topology of $\mathcal{R}$, $V_R$ is not geometrically local since it contains gates coupling qubits on opposite sides of $\mathcal{R}$. However, it is still $k$-local. Therefore, $W_1V_RW_2$ is a \fdqck that realizes $Q$ on a finite ring, where $Q$ is an arbitrary 1D QCA. {The elementary gates which make up the \fdqck are $S_i$, $V_i$, and $w$. In Appendix \ref{sec:qca_ex}, we derive these gates and demonstrate the general construction for several explicit examples.}

\subsection{Application to compactified higher-dimensional systems} \label{sec:qca_higherdim}

Having shown that all 1D QCA can be realized as a \fdqck, we now move on to higher dimensional QCA acting on $d$-dimensional lattices. For simplicity, we assume we have a simple hypercubic lattice structure although the construction should generalize to any translationally invariant QCA on a lattice with mirror symmetries. We also assume without loss of generality that the QCA has a unit range in all $d$ directions. It is known that a local Margolus form is not possible for two-dimensional QCA and higher \cite{Arrighi2008}. However, we can still obtain the desired results by applying the 1D Margolus representation, which was also used to understand the index theory of higher dimensional QCA \cite{Freedman2020}. To do this, we simply compactify the $d$-dimensional QCA $Q$ along all spatial dimensions except for one to obtain a quasi-1D chain of supersites $i$ each containing a number of sites that is extensive in the $d-1$ compactified dimensions. Since $Q$ has a unit range, it will spread operators contained in supersite $i$ only to supersites $i-1,i,i+1$ such that $Q$ can be viewed as a 1D QCA of unit range acting on the compactified system. Therefore, it may be written in the Maroglus representation of Eq.~\ref{eq:standard_form} where each solid vertical line now represents one supersite.

Applying the construction described above then automatically gives a \fdqck on the supersites realizing $Q$. However, we must be careful to confirm that the depth of the $k$-local circuit and the value of $k$ are both independent of system size in the compactified dimensions. This is clearly true of $V_R$, whose definition in Eq.~\ref{eq:vr} is unaffected by the compactification, and hence it is still $k$-local and finite-depth. Note that while the definition of $V_R$ depends on which dimensions we choose to compactify, $V_R$ is unaffected by whether we actually compactify those dimensions or not. What remains then is to check that the unitary $w$ used to define $W_1$ and $W_2$ can be realized as an \fdqck. Note that $W_1$ and $W_2$ each consists of a parallel application of $w$. Then, if $w$ can be realized as an \fdqck, so can $W_1$ and $W_2$. We show that this is the case by showing that $w$ is in fact a $(d-1)$-dimensional QCA and then using an inductive argument.

We now show that $w$ is locality-preserving as well as transitionally invariant in the compactified dimensions, \textit{i.e.} it is a $(d-1)$-dimensional QCA. Let $O$ be a local operator. Recall that $w=\bar{v}u$. Let $uOu^{-1} = \sum_k A^k \otimes B^k$ where $A^k$ and $B^k$ are operators supported only on the $\ell$- and $r$-dimensional Hilbert spaces that come out of $u$, respectively. Using the Schmidt Decomposition we can always choose $A^k$ ($B^k$) to come from a linearly independent set of operators acting on $\mathbb{C}^\ell$ ($\mathbb{C}^r$). Let $C^k=\bar{v}A^k\bar{v}^{-1}$ and $D^k=\bar{v}B^k\bar{v}^{-1}$. Note that since $v$ is a unitary map and since $A^k$ and $B^k$ are linearly independent, $C^k$ and $D^k$ are linearly independent as well. Then from linearity, we have,
\begin{align}\label{eq:wow}
    wOw^{-1} = \sum_k C^kD^k.
\end{align}
As we show below, the locality-preserving property of the original QCA $Q$ ensures that each $C^k$ and $D^k$ in this sum is localized around $O$ (meaning that they all are contained in a ball of finite radius centered around $O$), and thus $w$ is locality-preserving. 

Consider the action of $Q$ on $O$. In the following calculation, we will be more explicit with site indices; we write $O_{2i,2i+1}$ to represent the two (super)site operators $O$ acting on (super)sites $2i$ and $2i+1$. Then,
\begin{align}
    &QO_{2i,2i+1}Q^{-1} \nonumber \\
    &= \left( \prod_{j} v_{2j-1,2j} \right)\sum_k A^k_{2i}\otimes B_{2i+1}^k \left( \prod_{j} v^{-1}_{2j-1,2j} \right) \nonumber \\
    &= \sum_k \left(v_{2i-1,2i}A^k_{2i}v^{-1}_{2i-1,2i}\right) \otimes \left(v_{2i+1,2i+2}B^k_{2i+1}v^{-1}_{2i+1,2i+2}\right)  \nonumber \\
    &=\sum_k \bar{C}^k_{2i-1,2i} \otimes \bar{D}^k_{2i+1,2i+2}
    \label{eq:nonlocalop}
\end{align}
{In the first equation, we used the Margolus representation of $Q$ and conjugated $O$ by the $u$ operators.} The bars appearing on $C^k$ and $D^k$ again indicate spatial reversal since $A^k$ and $B^k$ have been conjugated by $v$ rather than $\bar{v}$. As is shown in Appendix \ref{app:w}, since $Q O_{2i,2i+1} Q^{-1}$ is contained around $O$, linear independence ensures that each $\bar{C}^k$ and $\bar{D}^k$ should be contained around $O$. This in turn shows that $w O w^{-1}$ in Eq.~\ref{eq:wow} is also contained around $O$, and hence $w$ is also locality-preserving in the $d-1$ compactified dimensions. A similar argument shows that $w$ is translationally invariant in the compactified directions if $Q$ is (see {Lemma \ref{lemma:trans} in} Appendix \ref{app:w} for details). In other words, $w$ is a $(d-1)$-dimensional QCA. 

We finish the proof using an inductive argument. We have already explicitly shown how to realize any 1-dimensional QCA as a \fdqck. Now suppose we can realize any $d$-dimensional QCA as a \fdqck. Then, given a $(d+1)$-dimensional QCA $Q$, we have shown how to prepare it using the \fdqck $V_R$ and the unitary $w$. Since $w$ is a $d$-dimensional QCA, we can by assumption realize it, and hence $Q$ itself, as an \fdqck. We note that $w$ is not an FDQC in general, so this inductive step is necessary. For example, if we consider the 2D QCA that shifts operators diagonally, then $w$ will be a 1D shift QCA, {as demonstrated in Appendix \ref{sec:qca_ex_2dshift}}.

\subsection{Symmetric QCA and SPT phases} \label{sec:qca_symm}

We now turn to the symmetry properties of the \fdqck{}'s constructed in the previous section. Suppose $Q$ is a symmetric QCA, meaning that it commutes with a \textit{global} symmetry $U(g) = u(g)^{\otimes N}$ for $g\in G$ (we discuss higher-form symmetries at the end). Note that, in contrast to a symmetric QC, where each gate in the circuit is individually symmetric, here we don't necessarily have a way to break $Q$ into smaller pieces, so we require only that $Q$ as a whole is symmetric, $[Q,U(g)]=0$. Let $U_{A}(g)$ ($U_{B}(g)$) denote $U(g)$ acting on the $A$ ($B$) subsystem as defined in the previous section. $U_A(g) U_B(g)$ clearly commutes with the swap operators $S_i$, since it is a tensor product of the same operator $u(g)$ on every site. Given that $Q_A$ is assumed to be symmetric, $V_i=Q_AS_iQ_A^{-1}$ commutes with $U_A(g) U_B(g)$ as well, and thus all gates in the circuit $V_R$ defined in Eq.~\ref{eq:vr} commute with $U_A(g) U_B(g)$ for all $g\in G$.
When we view the sites acted on by $V_R$ as the ring $\mathcal{R}$, $U_A(g) U_B(g)$ is just $U(g)$ acting on all sites in the ring, so it is just the global symmetry of the ring which we denote as $U_\mathcal{R}(g)$. For $c \le i \le d$, $V_i$ and $S_i$ have trivial support outside $\mathcal{R}$, and hence the fact that they commute with $U_A(g)U_B(g)$ readily shows that they also commute with $U_\mathcal{R}(g)$. Next, it follows from the results of \cite{Cirac2017} that $w$ commutes with $U(g)$ for all $g\in G$, see Appendix \ref{app:w} for detail. Therefore, $W_1 V_R W_2$ is a symmetric \fdqck whose gates commute with $U_\mathcal{R}(g)$ for all $g\in G$.

This result has implications for SPT phases. Observe that the \fdqck{}'s used to disentangle the SPT fixed-point states in Sec.~\ref{sec:spt} generate the same unitary operators as the non-symmetric FDQCs used to define the states in the first place, which we call SPT entanglers. That is, our constructions didn't just trivialize the fixed-point states, they achieved the stronger task of {expressing the SPT entanglers themselves as symmetric \fdqck{}'s}. This perspective allows us to apply our results on representing QCA as symmetric \fdqck{}'s to SPT phases. For example, fixed-point states for a large class of bosonic SPT phases (the ``in-cohomology'' phases) with global symmetry are given by the cocycle states defined in Ref.~\cite{Chen2013}. These states are in turn defined by FDQCs which commute with the global symmetry, but they are not symmetric circuits since the individual gates are not symmetric. Now, our construction allows these circuits to be written as \fdqck{}'s with symmetric gates. This shows that all in-cohomology SPT phases are trivial in the $k$-local scenario. {We demonstrate this idea for the example of the 1D cluster state in Appendix \ref{sec:qca_ex_cluster}.} 

Our construction can be applied to some beyond-cohomology SPT phases as well, although these are less well-understood at the Hamiltonian level. One example of a beyond-cohomology SPT phase in 4D with $\mathbb{Z}_2$ symmetry was given in Ref.~\cite{Fidkowski2020}, where a $\mathbb{Z}_2$-symmetric SPT entangler was also constructed. This entangler is a FDQC, so our construction can be applied to get a symmetric \fdqck. An interesting direction for further work requires extending our results to consider anti-unitary symmetries, \textit{i.e.} time reversal, which would allow us to address the $k$-local triviality of the 3D beyond-cohomology phase of Ref.~\cite{Burnell2014}.

Finally, we note that our approach immediately fails for higher-form symmetries. This is because the gates in $V_R$ commute only with $U_A(g)U_B(g)$ and not $U_{A}(g)$ or $U_{B}(g)$  individually. That is, the only symmetry operators that commute with $V_R$ are those that act in the same way on the front and back half of $\mathcal{R}$. But the full higher-form symmetry group includes operators that act differently on the two halves, as discussed in Sec.~\ref{sec:higherformSPT}. So $V_R$ does not commute with the full higher-form symmetry group. 

{
\section{Stability of SPT phases under generic $k$-local interactions} \label{sec:generic}

In the previous section, we have shown the existence of symmetric \fdqck{}'s which trivialize SPT phases. However, as these circuits are fine-tuned, this does not necessarily imply that SPT phases are unstable to \textit{generic} symmetric $k$-local interactions. Here, we argue that this is the case, meaning that generic symmetric $k$-local interactions will destroy SPT order. We first consider specific instances of interactions, and then later argue why the same results should hold generically.

We focus on the case of 1D SPT phases, but we expect that similar arguments will carry over to higher dimensions as well. We take $|\psi_{C}\rangle$ (Eq.~\ref{eq:1d_clus}) as the fixed-point state in a non-trivial SPT phase and construct perturbed states $|\psi\rangle$
by applying a short time-independent Hamiltonian time evolution, $|\psi\rangle=e^{-itH}|\psi_C\rangle$ for some $H$. We will consider several choices of $H$ subject to certain locality and symmetry constraints. By taking $t$ to be arbitrarily small, this gives a perturbed state that is arbitrarily close to the unperturbed state (in terms of fidelity per site). Therefore, we consider this a model of the effect of weak $k$-local noise on an SPT phase, which may or may not destroy the SPT order. This is in contrast to our exact disentangling \fdqck{}'s obtained in Sec.~\ref{sec:spt}, which require strong $k$-local interactions.

In order to diagnose the presence or absence of SPT order in the resulting state, we use the string order parameter \cite{Pollmann2012a}, which can be defined as,
\begin{equation} \label{eq:sop}
    S(a,b) = Z_aY_{a+1}\left(\prod_{i=a+2}^{b-1} X_i \right)Y_{b}Z_{b+1}.
\end{equation}
In the non-trivial SPT phase containing $|\psi_C\rangle$, the string order generically saturates to a non-zero value as its length is increased, whereas it goes to zero exponentially quickly in the trivial phase. In particular, we have $\langle \psi_C|S(a,b)|\psi_C\rangle=1$. To reduce the number of length scales, we evaluate the string order parameter over half of the system, \textit{i.e.} $S:= S(0,N/2)$, and study its behavior as a function of the system size $N$.

Let us first consider local, asymmetric noise, generated by the Hamiltonian $H^{(1)} = -\sum_i Z_i$. Since $Z_i$ anti-commutes with some symmetry generators of $|\psi_{C}\rangle$, this is not a symmetric Hamiltonian. Therefore, we expect that the time-evolved state will be in a trivial SPT phase and the string order will decay exponentially to zero with $N$. The state after an evolution time $t$ is,
\begin{equation}
    |\psi^{(1)}\rangle = \prod_{i=1}^N e^{itZ_i} |\psi_C\rangle .
\end{equation}
Define the subset of sites $\mathcal{I}=\{1,2,\dots,N/2\}$ which has the property that $Z_i$ anti-commutes with $S$ if $i\in\mathcal{I}$ and commutes otherwise. Then we can straightforwardly evaluate the string order in this perturbed state,
\begin{align}
    \langle \psi^{(1)}|S|\psi^{(1)}\rangle 
    &= \langle \psi_C| \left(\prod_{i=1}^N e^{-itZ_i}\right) S \left( \prod_{i=1}^N e^{itZ_i}\right) |\psi_C\rangle 
    \nonumber \\
    &=\langle \psi_C| \left( \prod_{i\in \mathcal{I}}e^{-2itZ_i}\right) S|\psi_C\rangle 
    \nonumber \\
    &= \langle \psi_C|\prod_{i\in \mathcal{I}}(\cos{2t}-i\sin{2t}\ Z_i) |\psi_C\rangle 
    \nonumber \\
    &= \gamma^{N/2}
    \nonumber \\
\end{align}
where $\gamma = \cos 2t$ and we used the facts that $S|\psi_{C}\rangle = |\psi_{C}\rangle$ and 
$
\langle \psi_C|\prod_{i\in \mathcal{S}} Z_i|\psi_C\rangle=0 
$
for any non-empty index set $\mathcal{S}$. 
Since $\gamma < 1$ for any non-zero $t<<1$, we see that the string order decays to zero exponentially quickly, indicating that the state $|\psi^{(1)}\rangle$ has trivial SPT order. 

Now, let us repeat the same calculation for a symmetric, local Hamiltonian. We choose the perturbation $H^{(2)} = -\sum_i Z_{i-1}Z_{i+1}$, which commutes with the $\mathbb{Z}_2\times\mathbb{Z}_2$ symmetry. The time-evolved state is,
\begin{equation}
    |\psi^{(2)}\rangle = \prod_{i=1}^N e^{itZ_{i-1}Z_{i+1}}|\psi_C\rangle.
\end{equation}
The calculation of string order is largely the same. The key difference is that the set of sites $i$ for which $Z_{i-1}Z_{i+1}$ anti-commutes with $S$ is finite, containing only sites $i=0,1,N/2,N/2+1$. Therefore, we find that,
\begin{equation}
    \langle \psi^{(2)}|S|\psi^{(2)}\rangle = \gamma^4,
\end{equation}
so the string order is a non-zero constant independent of $N$, indicating that $|\psi^{(2)}_C\rangle$ has non-trivial SPT order as expected. 

Finally, we consider the case of $k$-local symmetric perturbations. Based on the disentangling circuits we constructed in the previous sections, we expect that this will trivialize the SPT order. Let us define a set $\mathcal{A}$ consisting of random pairings of sites $(i,j)$ such that every site is contained in exactly one pair, and $i$ and $j$ are either both even or both odd. Then, we consider the 2-local symmetric perturbation $H^{(3)} = -\sum_{(i,j)\in\mathcal{A}} Z_i Z_j$. We could also consider the case where every qubit interacts pairwise with every other qubit, but the calculation is greatly simplified by assuming each qubit only interacts with one other. Note also that this Hamiltonian has constant energy density despite being long-range interacting. As before, we consider the state,
\begin{equation}
    |\psi^{(3)}\rangle = \prod_{(i,j)\in\mathcal{A}} e^{itZ_i Z_j} |\psi_C\rangle.
\end{equation}
Now split $\mathcal{A}$ into two subsets $\mathcal{A}_e$ and $\mathcal{A}_o$ such that $(i,j)\in\mathcal{A}_e$ if $i$ and $j$ are either both in $\mathcal{I}$ or neither are, while $(i,j)\in\mathcal{A}_o$ if one if $i,j$ is in $\mathcal{I}$ and the other is not. Then, $Z_iZ_j$ anti-commutes with $S$ if and only if $(i,j)\in\mathcal{A}_o$. Following the above calculations, we find,
\begin{equation}
    \langle \psi^{(3)}|S|\psi^{(3)}\rangle = \gamma^{|A_o|}
\end{equation}
where $|\mathcal{A}_o|$ is the number of pairs in $\mathcal{A}_o$. Given a random pairing of sites described by the set $\mathcal{A}$, if we take any site $i\in\mathcal{I}$, its partner will be in $\mathcal{I}$ with probability $\sim 1/2$ since $\mathcal{I}$ contains half of the lattice sites. Therefore, we expect for typical pairings that $|\mathcal{A}_o|\approx N/4$ (as $|\mathcal{A}|=N/2$), so the string order decays exponentially with $N$, indicating trivial SPT order.

\begin{figure}
    \centering
    \includegraphics[width=\linewidth]{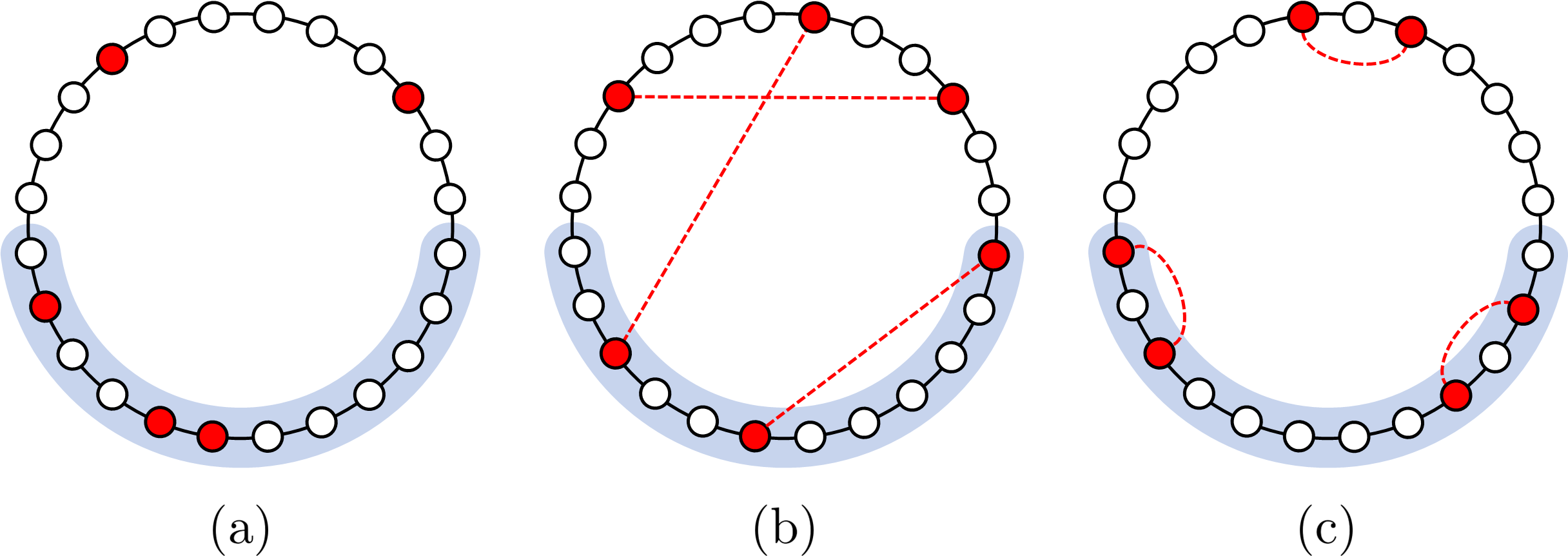}
    \caption{Illustration of typical components of the perturbed states for (a) random, (b) symmetric $k$-local, and (c) symmetric local perturbations. Each red circle denotes the application of one $Z$ operator, and the dashed lines connecting two circles represent symmetric pairs of $Z$'s. The shaded regions indicate the region over which the string order is evaluated. In a typical configuration, (c) will likely have an even number of $Z$'s in the shaded region, whereas (a) and (b) have no parity bias.}
    \label{fig:perturbed_clusters}
\end{figure}

The above analysis makes it clear why symmetric $k$-local perturbations destroy SPT order. Consider the two states $|\psi^{(1)}\rangle$ and $|\psi^{(3)}\rangle$. Each is a superposition over states of the form $\propto \prod_{i\in S}Z_i|\psi_C\rangle$ for some index sets $\mathcal{S}$. The only significant difference between the two states is that, in the case of $|\psi^{(3)}\rangle$, the global symmetry constraint requires $\prod_{i\in S}Z_i$ to contain an even number of $Z$'s on both sublattices. However, the string order parameter, which is evaluated only over half of the lattice sites, does not see this global constraint; there is a high probability for an odd number of $Z$'s to be applied to the region where the string order parameter acts, see Fig.~\ref{fig:perturbed_clusters}. This anti-commutation of the perturbations with the string order parameter leads to destructive interference which causes it to decay. In contrast, in $|\psi^{(2)}\rangle$, the $Z$'s always appear in pairs separated by a short distance, such that only those pairs which straddle the boundary of $\mathcal{I}$ will anti-commutes with $S$, which gives only a finite correction to the string order, see Fig.~\ref{fig:perturbed_clusters}. In other words, the $k$-local symmetric perturbations allow you to freely violate the symmetry in any local region, and this is what leads to the breakdown of the SPT order. Indeed, the value of the string order parameter within any region in which the symmetry is violated necessarily goes to zero \cite{Perez-Garcia2008}. From this reasoning, it is clear that any generic $k$-local perturbation will similarly destroy SPT order. 

On the other hand, if the $k$-local perturbation is \textit{locally symmetric}, meaning that the perturbation has a form like $O_iO'_j$ where $O_i$ and $O'_j$ are symmetric local operators, then this does not violate the symmetry locally, and we expect the SPT order will be robust to such perturbations. This symmetry restriction is the same as the symmetry restriction for higher-form SPTs as discussed in Sec.~\ref{sec:higherformSPT}.
}

\subsection{Instability of SPT states in monitored random circuits}\label{sec:randomcirc}

{
We have argued that SPT order is unstable to generic $k$-local symmetric perturbations, but it is stable to $k$-local perturbations which are locally symmetric. In this section, we give further numerical evidence of these claims. We consider monitored quantum circuits, which involve both unitary gates and projective measurements that are randomly applied to the state with some probability (for a review, see \cite{fisheretalreview}). In general, it has been observed that these elements compete with each other, driving the late-time state of the evolution to different regimes of behavior depending on their relative frequency. In particular, certain symmetric monitored random circuits have been studied which can sustain SPT order within a certain range of the circuit parameters \cite{LavasanietalSPT, friedmanetaladaptive}. We study the implication of our results for the stability of SPT order in this context and use this to give further evidence on which $k$-local circuits can and cannot trivialize SPT phases. 
}

Consider arranging $N$ qubits, initialized in the $|+\rangle ^ {\otimes N}$ state, on a 1D ring and applying the following quantum process: at each step, with probability $p$ a random two-qubit unitary $U$ is applied to the system or, with probability $1-p$ a qubit $i$ is chosen uniformly at random and $g_i\equiv Z_{i-1}X_iZ_{i+1}$ is measured. The latter tends to drive the state towards the 1D cluster state which satisfies $g_i|\psi_C\rangle=|\psi_C\rangle$, while the former tends to drive it away. {A time step is defined to consist of $N$ consecutive steps.} In the following, we consider four different ensembles of two-qubit unitaries: a) all two-qubit geometrically local Clifford unitaries, b) two-qubit geometrically local Clifford unitaries that respect the $\mathbb{Z}_2\times \mathbb{Z}_2$ symmetry generated by $X_{\mathrm{odd}}$ and $X_{\mathrm{even}}$,
c) $2$-local Clifford unitaries that respect the same $\mathbb{Z}_2\times \mathbb{Z}_2$ symmetry and finally, d) $2$-local Clifford unitaries that respect the $\mathbb{Z}_2^{N}$ symmetry generated by $X_i$ for $i=1,\cdots,N$, i.e., are diagonal in the local $X$-basis and are therefore locally symmetric. In the case of geometrically local unitaries, a site $i$ is chosen at random, and a two-qubit unitary, chosen randomly from the appropriate ensemble, is applied to qubits $i$ and $i+1$. As for $2$-local unitaries, two different sites $i$ and $j$ are chosen randomly and then a random unitary from the appropriate ensemble is applied to them. We are interested in the late-time states of this family of random circuits, {which we take to be the quantum state of the circuit after $T=N$ time steps.} 

At $p=0$, the circuit consists of only $g_i$ stabilizer measurements. Therefore, the late-time state of any realization of the random circuit would be an SPT state. The SPT nature of this state can be probed by the non-local analog of Edwards–Anderson glass-order parameter\cite{bahri2015localization,sang2021measurement, LavasanietalSPT},
\begin{align}
    s=\frac{2}{N(N-1)}\sum_{a<b}S(a,b)^2,
\end{align}
with $S(a,b)$ being the SPT string order parameter defined in Eq.~\ref{eq:sop}.
$s$ is equal to $1$ for $|\psi_C\rangle$. As described earlier, SPT states are characterized by $s>0$, while for trivial states or random states, $s=0$ in the thermodynamic limit. We are interested in $\bar s$, which is $s$ averaged over random circuit realizations. We note that $s$ is not an {\it experimentally} accessible quantity, but it is certainly accessible in simulations, which is sufficient for our present purposes. 

\begin{figure}
    \centering
    \subfigure[]{{\label{fig:local_u}\includegraphics[width=0.49\linewidth]{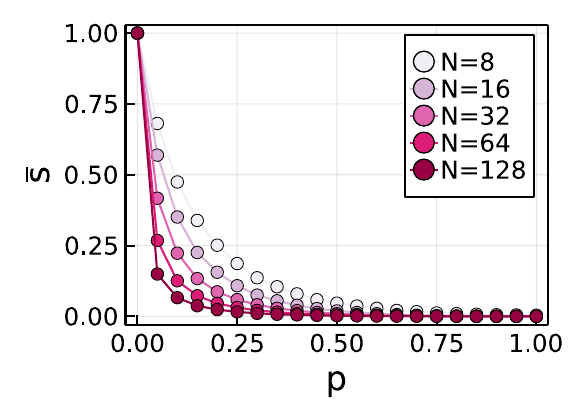}}}
    \subfigure[]{{\label{fig:local_sym_u}\includegraphics[width=0.49\linewidth]{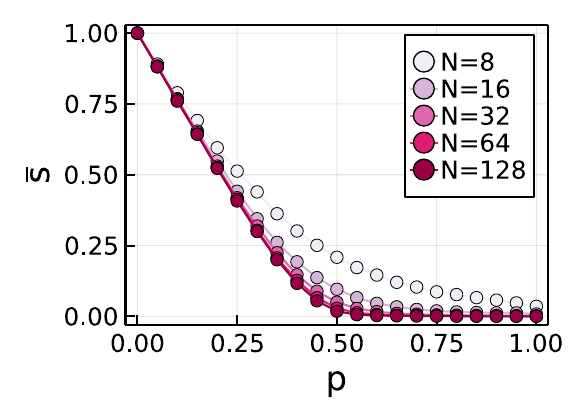}}}
    \subfigure[]{{\label{fig:2local_sym_u}\includegraphics[width=0.49\linewidth]{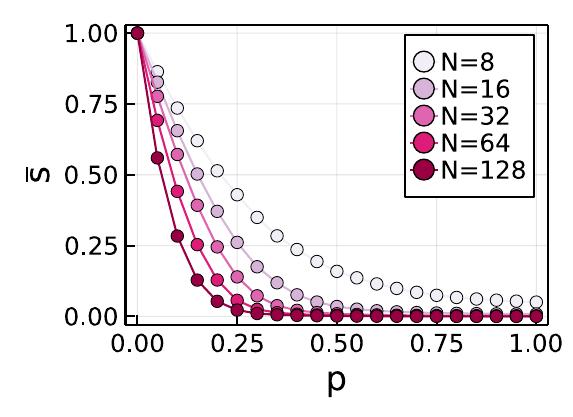}}}
    \subfigure[]{{\label{fig:2local_locally_sym_u}\includegraphics[width=0.49\linewidth]{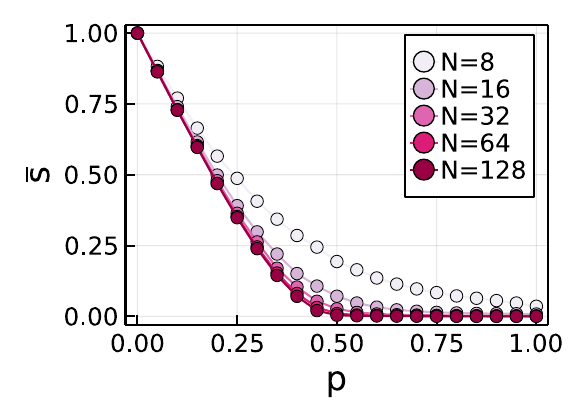}}}
    \caption{Averaged string order parameter $\bar{s}$ versus $p$ for monitored random circuits where the unitary gates are chosen randomly from a) local two-qubit Clifford unitaries, b) local two-qubit $\mathbb{Z}_2 \times \mathbb{Z}_2$ symmetric Clifford unitaries, c) 2-local $\mathbb{Z}_2 \times \mathbb{Z}_2$ symmetric Clifford unitaries, and d) 2-local $\mathbb{Z}_2^N$ symmetric Clifford unitaries.}
    \label{fig:sbar_plots}
\end{figure}
Fig.~\ref{fig:sbar_plots} shows $\bar{s}$ as a function of $p$ for each of the unitary ensembles described above. Fig.~\ref{fig:local_u} corresponds to the random circuit where the unitary gates are chosen from the ensemble of local two-qubit Clifford gates without imposing any symmetry restriction. As expected, the SPT structure in the late time state vanishes for any $p>0$ as the symmetry is violated. And on the other hand, if we restrict the local unitary gates to respect the $\mathbb{Z}_2 \times \mathbb{Z}_2$ symmetry, the SPT structure survives up to finite $p_c>0$, below which $\bar{s}$ saturates to a finite value, as is shown in Fig.~\ref{fig:local_sym_u}. Interestingly, when the unitary gates are chosen from the set of $2$-local Clifford unitaries that respect the $\mathbb{Z}_2\times \mathbb{Z}_2$ symmetry, the SPT structure vanishes again for any $p>0$, as illustrated in Fig.~\ref{fig:2local_sym_u}. This is consistent with our findings that SPT states can be trivialized by $k$-local symmetric unitary gates. Lastly, Fig.~\ref{fig:2local_locally_sym_u} corresponds to the circuit where the unitary gates are chosen from the highly constrained ensemble of 2-local Clifford unitaries that map $X_i$ to itself for all $i$, \textit{i.e.} they are locally symmetric.
Despite the fact that this set includes long-range entangling gates, we see that the SPT structure survives up to a finite $p>0$, which is consistent with our intuition that $k$-local gates that are locally symmetric are not much more powerful in terms of disentangling the SPT structure compared to geometrically local symmetric gates. {In Appendix \ref{app:spt_numerics} we provide further analysis of the numerical data shown in Fig.~\ref{fig:sbar_plots}}

\section{Discussion}
\label{sec:discussion}

We have constructed explicit finite-depth circuits consisting of (symmetric) $k$-local gates to create fixed-point SPT states and to realize all QCA. The circuits imply that the classification of SPT phases (with global or subsystem symmetries) and the classification of QCA collapse to a single trivial phase in the presence of $k$-local interactions. 
{This addresses {\it worst case stability} - whether there exists an FDQC of $k$-local gates that can trivialize a given state - and suffices to show that these phases (strictly speaking) do not exist in the presence of $k$-local interactions. We also addressed the case of {\it typical case stability} by giving analytical and numerical evidence that SPT order is also fragile in the presence of generic $k$-local symmetric perturbations.
}
We note that a key ingredient in our explicit circuits was to `fold' the system in such a way that it resembled a stack of the system with its inverted self. This naturally assumes some sort of mirror symmetry is present in the lattice. Therefore, it is natural to ask whether our results can be applied to systems on more general lattices and to what extent translational invariance can be relaxed.

We remark that SPT phases can be used as resources for measurement-based quantum computation (MBQC), see \cite{Wei2018a} for a review. That is, given a ground state in a suitable SPT phase, one can perform measurements on that ground state in order to simulate a quantum computation. While the computational capability of a state as a resource for MBQC is stable to local symmetric perturbations, our results seem to imply that it may be unstable to $k$-local symmetric perturbations in the case of global or subsystem symmetries. Indeed, in Ref.~\cite{Raussendorf2019}, the authors explicitly identified a symmetric $k$-local interaction whose presence would invalidate the strategy that was used to prove MBQC universality. In contrast, fault-tolerant MBQC can be achieved using SPT phases with higher-form symmetries \cite{Raussendorf2006}, which appear to be stable even under symmetric $k$-local interactions.

While we have focused on bosonic systems, we believe that many of our arguments should apply to fermionic systems as well. Indeed, our physical arguments for the triviality of SPT phases based on their boundary physics and invertibility should carry over {\it mutatis mutandis}. Specifically, the free fermion SPT phases (i.e., topological insulators and superconductors) are characterized by non-trivial boundary states, which are manifestly rendered trivial if one can couple distinct boundaries together via $k$-local interactions. Meanwhile, \cite{WangSenthil} showed that the set of fermionic SPT phases in three dimensions, protected by some combination of antiunitary symmetries and charge conservation, is exhausted by combinations of the free fermion SPTs and bosonic SPTs - and bosonic SPTs we have constructively shown to be trivializable. This suggests that these fermion SPTs should also be trivializable by $k$-local interactions. Similarly, we expect that many, if not all, fermionic QCA should become trivial in the $k$-local setting \cite{Fidkowski2019,Piroli2021a}. However, establishing rigorous results on fermionic SPTs and QCA is left to future work. 

In contrast to SPT states and QCA, intrinsic topological order is known to be stable under finite-depth $k$-local circuits \cite{Aharonov2018}, meaning that a non-trivial topological phase will not become trivial. However, this doesn't mean that the classification of intrinsic topological phases will remain unchanged. Indeed, in the $k$-local scenario, it is possible that the braiding of topological excitations becomes ill-defined, so some topological phases which differ only by braiding statistics may become the same phase. 
Also, in symmetry-enriched topological (SET) phases, the topological order is stable but the nontrivial symmetry fractionalization pattern of the SET may not be stable under $k$-local circuits. For instance, certain symmetry fractionalization patterns can be canceled by stacking with SPT states \cite{Barkeshli2019}. Since we know that SPT states can be prepared with finite-depth $k$-local circuits, such symmetry fractionalization patterns are not stable to finite-depth $k$-local circuits. A thorough investigation of the stability of intrinsic topological order and symmetry fractionalization is left for future work. 

\begin{figure}
    \centering   
    \subfigure[]{{\label{fig:toric_code_stack_ERG}}\includegraphics[width=0.85\linewidth]{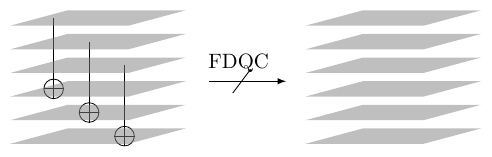}} \\ \vspace{5mm}
    \subfigure[]{{\label{fig:Xcube_ERG}}\includegraphics[width=0.9\linewidth]{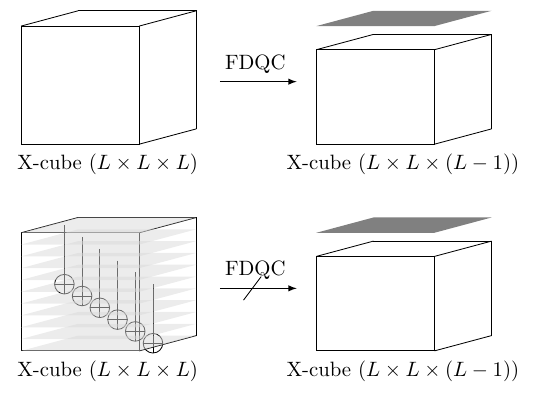}}
    \caption{(a) Action of a $k$-local circuit on a stack of toric codes, in general, leads to a model that is no longer a finite-depth local quantum circuit (FDQC) equivalent to a stack. (b) Top: X-cube model is foliated i.e., there exists a finite-depth local quantum circuit under which an $L\times L\times L$ X-cube model maps to an X-cube model of dimensions $L\times L\times (L-1)$ and a layer of toric code. Bottom: Action of a $k$-local circuit on the X-cube model, in general, leads to a model that is no longer finite-depth local quantum circuit (FDQC) equivalent to a foliated model.}
    \label{fig:stackERG}
\end{figure}

Another class of phases for which the action of $k$-local circuits is interesting is fracton phases (for reviews, see \cite{fractonarcmp, pretkoreview, radzihovskyreview}). The defining property of fracton phases is the restricted mobility of the excitations, which comes about because the non-trivial excitations are not locally creatable, but instead, arise at the `corners' of extended operators. For instance, in the X-cube model \cite{VijayHaahFu} the fractons arise at the corners of membrane operators. Since these excitations cannot be created (or destroyed) by any few-body operator, they cannot be moved by $k$-local perturbations (without creating additional excitations), and thus the restricted mobility survives. Nonetheless, not all properties of fracton phases are unchanged. For instance, the geometric nature of entanglement can be modified under $k$-local gates. To illustrate this,  we specialize to {\it gapped} fracton phases, which are characterized by nonlocal entanglement with geometric rigidity \cite{decipher, maetal, schmitzetal, bernevigentanglement}. We expect that this nonlocal entanglement survives under $k$-local gates, much like the corresponding entanglement structure in phases with intrinsic topological order, but (we will argue) the geometric structure does not survive. For instance, consider a stack of 2D toric codes. Each copy of 2D toric code has a topological order that is stable under $k$-local gates that couple degrees of freedom in that copy alone. However, under a $k$-local circuit that couples different copies of toric codes, the foliation structure of the stack can be lost and no longer recoverable using a geometrically local FDQC; see Fig.~\ref{fig:toric_code_stack_ERG}. This argument extends to foliated fracton codes (see e.g. \cite{shirleyfoliated}). Under a finite-depth geometrically local quantum circuit (FDQC), a stack of two-dimensional topological codes, which is said to form the foliations of the fracton order, can be extracted. For extracting one layer of a two-dimensional model alone, the `exfoliating' circuit applies local disentangling operations on the boundary of the fracton order. However, a $k$-local circuit can entangle the extracted foliations back into the fracton order in a way that it can no longer be disentanglable via an FDQC. In this sense, the structure of the entanglement has changed i.e., there is no ex-foliation via a geometrically-local FDQC. We illustrate this for the example of the X-cube model in Fig.~\ref{fig:Xcube_ERG}. Similar considerations apply to the more general notion of bifurcating entanglement renormalization which has been explored for fracton orders~\cite{Haah_2014,Dua_2020} and SPT states~\cite{Miguel_2021}. A detailed exploration of fracton phases under $k$-local circuits is a promising topic for future work. 

Our results have a bearing on the classification of chiral Floquet phases of matter.
The non-trivial nature of these models is based on the fact that the effect of the Floquet dynamics at the boundary of an open system cannot be realized by an FDQC. The boundary dynamics of the models in \cite{Po2016} are given exactly by 1D bosonic QCA. Hence, our results show that the boundaries are no longer anomalous if $k$-local interactions are allowed, as 1D QCA are all \fdqck. On the other hand, the radical Floquet phases introduced in \cite{Poetal} have boundary dynamics described by a shift of a Majorana fermion which has a fractional index. As such an index is beyond bosonic QCA, and the bulk model consists of bosonic degrees of freedom, it is plausible that these phases remain robust to $k$-local interactions unless ancillary fermionic degrees of freedom are added.

Recently, Floquet codes have also been discussed in the literature~\cite{Hastings_2021}. Such codes are defined by a series of instantaneous codes, each of which is an instance of intrinsic topological order. Hence, it is natural to expect that such codes are stable to $k$-local noise. It could also be interesting to discuss symmetry-enriched Floquet codes, bearing in mind the subtleties associated with defining symmetries for Floquet unitaries \cite{RegnaultFloquet}

Our work can also be extended by modifying the definition of a locality-preserving unitary. One can weaken the constraint of locality-preserving by allowing exponential tails in the definition of a local operator. We conjecture that there exist $k$-local circuits also for such \textit{approximate locality-preserving unitaries}~\cite{Ranard_2022}, as would be realized by Hamiltonian time evolution. As a concrete application, we expect that invertible chiral states (like integer quantum Hall states) can be trivialized by $k$-local Hamiltonians, using constructions similar to the ones presented in our paper, although since we expect that chiral phases cannot be captured by zero-correlation length models, there may not be a nicely behaved truncation to $k$-local circuits. Investigations of approximately locality-preserving unitaries would also connect to e.g. the literature on state preparation with long-range Hamiltonians \cite{LucasGorshkov}. Lastly, we can consider a $k$-locality-preserving unitary ($\mathrm{QCA}_k$) which maps local operators to $k$-local operators acting on at most $k$ qubits. While we have shown that every QCA a \fdqck, this may not be true for $\mathrm{QCA}_k$.

\begin{acknowledgments}
We thank Daniel Bulmash, Tyler Ellison, Mike Hermele, and Dominic Williamson for discussions. We thank Jeongwan Haah, Tyler Ellison, Drew Potter, and Carolyn Zhang for their feedback on the manuscript. 
Work by RN was supported by the U.S. Department of Energy (DOE), Office of Science, Basic Energy Sciences (BES) under Award \# DE-SC0021346. This work was begun during a visit to the Aspen Center for Physics (DTS, AD, and RN). The Aspen Center for Physics is supported by the National Science Foundation grant PHY-1607611. AD and DTS are supported by the Simons Foundation through the collaboration on Ultra-Quantum Matter (651440, DTS; 651438, AD). AD is supported by the Institute for Quantum Information and Matter, an NSF Physics Frontiers Center (PHY-1733907). AD, AL, and RN acknowledge the KITP program on Quantum Many-Body Dynamics and Noisy Intermediate-Scale Quantum Systems, where part of the work was completed. The KITP is supported by the National Science Foundation under Grant No. NSF PHY-1748958. The authors acknowledge the University of Maryland supercomputing resources (http://hpcc.umd.edu) made available for conducting the research reported in this paper.

\end{acknowledgments}

\bibliography{biblio.bib}

\appendix

\section{Finite time protocols} \label{app:finitetime}

In this section, we discuss $k$-local protocols which are not strictly finite-depth but can nonetheless be generated in finite time by a $k$-local Hamiltonian. The advantage of the circuits constructed here is that all interactions are geometrically local, except for one special qubit which can interact with every other qubit. That is, we add in a ``one-to-all'' interaction. While the special qubit is involved in the circuit used to create the SPT-ordered state, at the end of the circuit it remains decoupled from the rest of the spins. This can be viewed as modeling the scenario of a central spin problem, or a cavity QED setup in which atoms in a cavity all couple to a cavity mode. As a tradeoff, these circuits are no longer finite-depth, but they can nonetheless be implemented via the finite time evolution of a symmetric $k$-local Hamiltonian.

As a first example, we again consider the 1D cluster state $|\psi_C\rangle$. Suppose we have an even number $N$ of qubits on a ring indexed as $i=1,\dots N$, plus an additional qubit indexed by $i=0$ which lives in the middle of the ring. Suppose in addition that this qubit transforms under the symmetry as an even qubit, so that $X_{\mathrm{even}} = \prod_{i=0}^{N/2}X_{2i}$, and consider the initial state $|\psi_C\rangle\otimes |+\rangle_0$. Then the following $k$-local gates are symmetric,
\begin{equation}
    V_i = CZ_{0,2i-1}CZ_{2i-1,2i}CZ_{2i,2i+1}CZ_{2i+1,0}.
\end{equation}
This gate is pictured in Fig.~\ref{fig:onetoall}. Then we have,
\begin{equation}
    \left( \prod_{i=1}^{N/2} V_i \right)|\psi_C\rangle \otimes |+\rangle_0= |+\rangle^{\otimes N}\otimes |+\rangle_0,
\end{equation}
so this disentangles the cluster state. We remark that, although there is intermediate entanglement between the central qubit and the ring, in the end, the central qubit remains in a product state. This circuit is not strictly finite-depth since every gate acts on the central qubit 0, so if we require layers to have non-overlapping gates it would require a linear number of layers. However, since all gates commute, they can be applied in finite time. 

\begin{figure}
    \centering
    \includegraphics[scale=0.33]{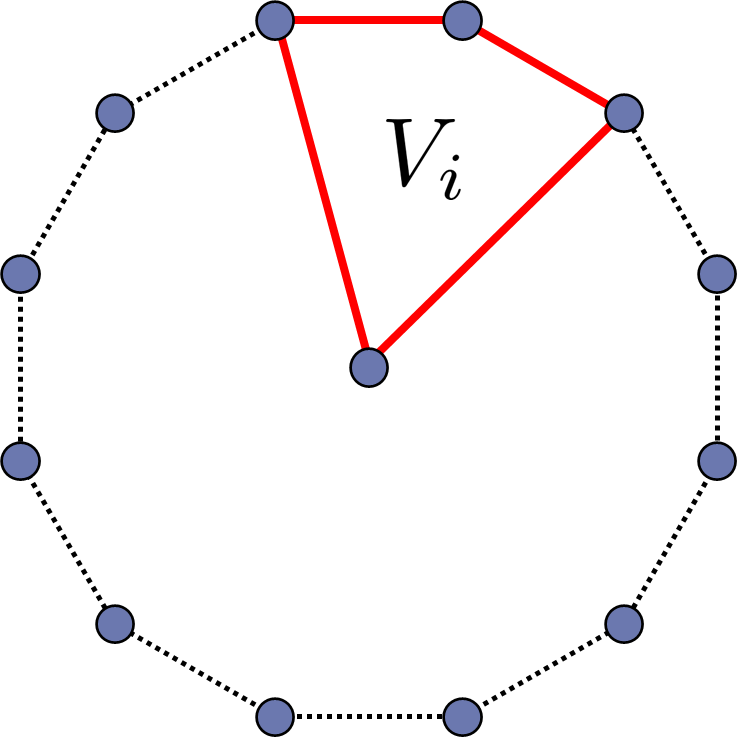}
    \hspace{5mm}
    \includegraphics[scale=0.35]{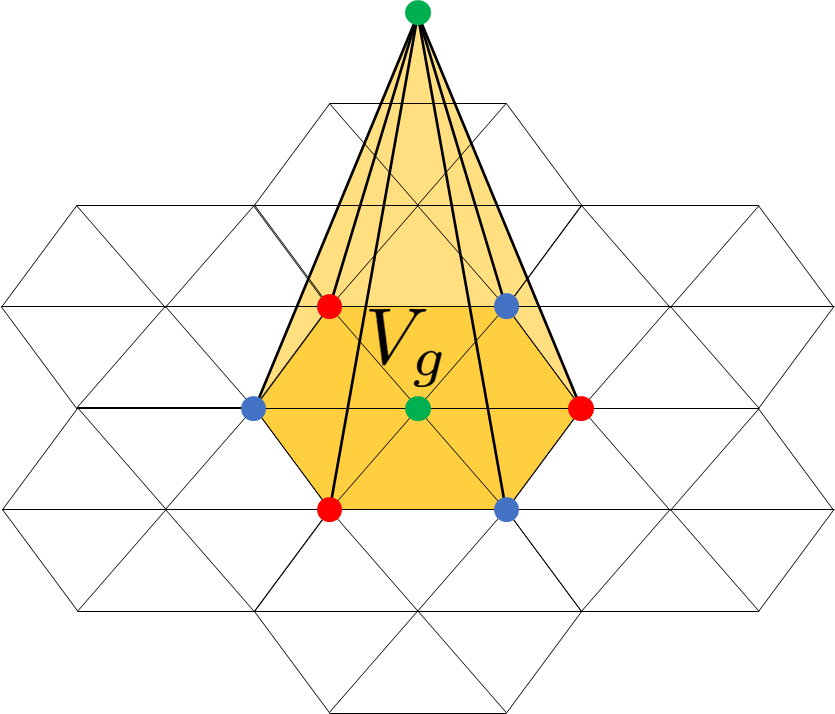}
    \caption{Left: One gate in the one-to-all circuit for preparing the 1D cluster state, where each solid line represents a $CZ$ gate. Right: One gate in the one-to-all circuit for the 2D hypergraph state where each colored triangular face represents a $CCZ$ gate. In both cases, the product of all gates results in all unitaries involving the special qubit canceling pairwise.}
    \label{fig:onetoall}
\end{figure}

The same procedure works for 2D SPTs. Consider again the hypergraph state $|\psi_H\rangle$ on a closed 2D surface. We again add a special qubit 0 which we assume transforms like a qubit on a green-colored site. Let $G$ denote the set of all green qubits on the 2D surface (not including the special qubit). Then, for every $g\in G$ let $g_j$ for $j=1,\dots,6$ denote the 6 qubits neighboring $g$ which alternate between red and blue. Then define the $k$-local symmetric gates,
\begin{equation}
    V_g = \prod_{j=1}^6 CCZ_{g,g_j,g_{j+1}} CCZ_{0,g_j,g_{j+1}}
\end{equation}
The gate is pictured in Fig.~\ref{fig:onetoall}. Then we have,
\begin{equation}
    \left(\prod_{g\in G} V_g\right) |\psi_H\rangle\otimes|+\rangle_0 = |+\rangle^{\otimes N}\otimes|+\rangle_0.
\end{equation}

This argument extends to all cocycle states in all dimensions as before. We remark that, in all cases, it is important that the special qubit transforms in a certain way under the symmetry. If the symmetry did not act on the special qubit, the $k$-local gates we defined would not be symmetric.

This construction can be understood in terms of the path integral representation of SPT order. By viewing the spatial manifold on which the SPT is defined as the boundary of some space-time in one higher dimension, one can construct the SPT on the boundary using a product of local symmetric gates in the bulk \cite{Chen2013}. The circuits we described are exactly of this form, where the bulk consists of a single ancilla qubit that couples to all qubits on the boundary. These same ideas have been used to define quantum pumps that pump a $d$-dimensional SPT from a $(d+1)$-dimensional bulk to the boundary. Such pumps have been constructed for general SPTs \cite{Potter2017,Roy2017,Tantivasadakarn2022,Shiozaki2022} and consist of symmetric gates in the bulk.

We remark that the notion of finite-time preparation is likely too powerful when it comes to the classification of phases. Indeed, any stabilizer state such as the GHZ state and the toric code ground state is equivalent to a graph state up to local unitaries \cite{VandenNest2004}. A graph state is any state that can be prepared from a product state of all $|+\rangle$ states using $CZ$ gates between pairs of qubits. Since all of these $CZ$'s commute, the graph state can be prepared in finite time using an Ising-type Hamiltonian. However, an important caveat is that, in the case of the graph states which are equivalent to the GHZ and toric code states, there are qubits that must interact with a number of other qubits that is extensive in linear system size \cite{Liao2021}, which is somewhat unphysical. For example, if we 
were to impose the physical constraint that the total strength of interactions involving any one qubit is finite in the thermodynamic limit, \textit{i.e.} a finite energy density, we must scale down the interaction strength of each Ising term accordingly. This has the consequence of requiring an interaction time that grows as the linear system size, which is consistent with the linear circuit depth. Conversely, if we do not scale the interactions down in this way, then we don't have a sensible thermodynamic limit - the energy density diverges as we make system size large, indicating that our effective low energy description in terms of an Ising Hamiltonian ceases to be a good approximation. 

{
\section{Examples of $k$-local circuits for QCA} \label{sec:qca_ex}

In this section, we illustrate our general construction of \fdqck{}'s for QCA with a number of examples.

\subsection{1D shift QCA} \label{sec:qca_ex_shift}

The 1D shift QCA which translates all sites to the right by one, can be represented in the standard form with $d=2$, $\ell=1$, and $r=4$. The $u$ and $v$ operators are defined as, 
\begin{equation}
u= \sum_{ij} |ij\rangle (\langle i|\otimes\langle j|),
\end{equation}
and,
\begin{equation}
v = \sum_{ij} (|i\rangle\otimes |j\rangle) \langle ij|,
\end{equation}
where $u$ maps from a pair of two-dimensional Hilbert spaces each spanned by the states $|0\rangle$ and $|1\rangle$ into a four-dimensional Hilbert space spanned by $|00\rangle$, $|01\rangle$, $|10\rangle$, and $|11\rangle$, and similarly for $v$. Graphically, we can draw these operators as,
\begin{equation} \label{eq:shift_standard_form}
    \includegraphics[scale=0.22]{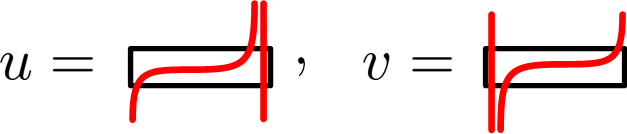}.
\end{equation}
One can see how these operators implement the shift QCA by drawing the whole circuit,
\begin{equation} \label{eq:shift_qca_standard}
    \includegraphics[scale=0.44]{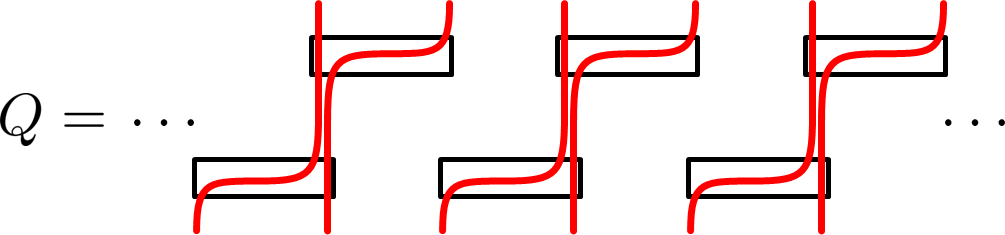}
\end{equation}
Following the lines, it is clear that every site is translated to the right by one.

The spatial inversions of $u$ and $v$ are obtained by swapping the left and right input and output Hilbert spaces, giving $\bar{u} = \sum_{ij} |ij\rangle (\langle j|\otimes\langle i|)$ and $\bar{v} = \sum_{ij} (|j\rangle\otimes |i\rangle) \langle ij|$.
From this we compute, 
\begin{equation}
w=\bar{v}u = \sum_{ij} (|j\rangle\otimes |i\rangle)( \langle i|\otimes \langle j|)=\mathrm{SWAP},
\end{equation}
and $\bar{w}=w$. Next, note that the circuit $V_R$ defined in Eq.~\ref{eq:vr} consists of two types of gates, the simple SWAP gate $S_i$ and the gate $V_i = Q_A^{-1} S_i Q_A$, which in the present case is also a SWAP gate. 

Finally, we apply the gates $S_i$, $V_i$, and $w$, which are all SWAP gates, in the order described by Eq.~\ref{eq:qcafdqck}. This gives a depth-4 \fdqck composed of 2-local SWAP gates realizing the shift QCA, as shown in Fig.~\ref{fig:shift_qca_general}. The resulting \fdqck is similar to that shown in Fig.~\ref{fig:shift_qca}, although somewhat more complicated as a consequence of the more general construction it comes from.

\begin{figure}
    \centering
    \includegraphics[scale=0.44]{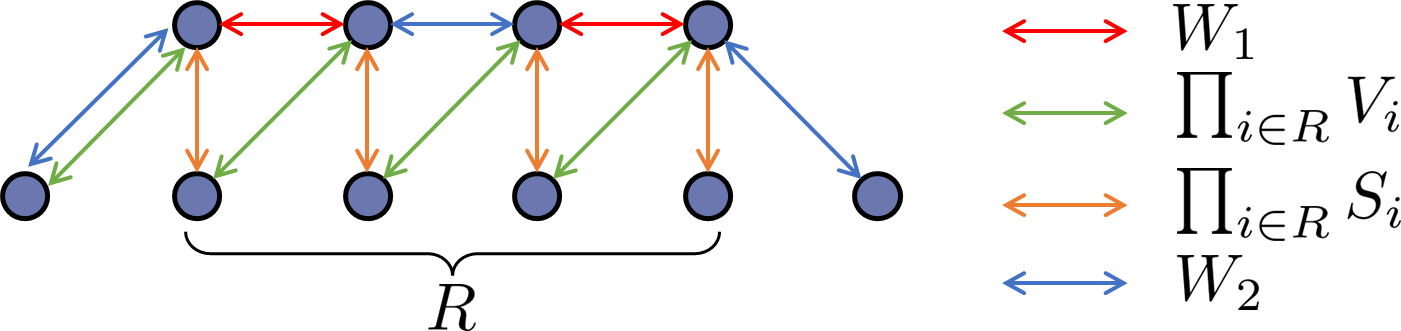}
    \caption{\fdqck for the 1D shift QCA resulting from the general contruction of Sec.~\ref{sec:qca}. The lower (upper) row of sites belongs to the $A$ ($B$) system. The SWAP gates are applied in the order red, green, yellow, blue, which have the combined effect of shifting all sites by one counterclockwise. Notice the similar trapezoidal arrangement of sites compared to the general structure pictured in Fig.~\ref{fig:vr}.}
    \label{fig:shift_qca_general}
\end{figure}

\subsection{1D cluster state} \label{sec:qca_ex_cluster}

As discussed in Sec.~\ref{sec:qca_symm}, we can also apply our construction to the FDQC's which create SPT fixed-point states. We illustrate this for the simple example of the 1D cluster state. The FDQC which constructs the 1D cluster state is given in Eq.~\ref{eq:1d_clus}. Observe that this circuit commutes with the $\mathbb{Z}_2\times\mathbb{Z}_2$ symmetry generated by $X_{\mathrm{odd}}$ and $X_{\mathrm{even}}$ as a whole, but the individual gates which compose the circuit do not commute with the symmetry. Therefore, this is not a symmetric FDQC. Now, we can apply our general construction with $Q=\prod_i CZ_{i,i+1}$ to obtain a symmetric \fdqck that implements the same operator. 

The operators $u$ and $v$ are both equal to $CZ$. Note that $CZ$ is symmetric under spatial inversion. Then, we have $w=I$, and the \fdqck is simply given by $V_R$ (Eq.~\ref{eq:vr}). As before, $V_R$ is composed of SWAP gates $S_i$ and the gates $V_i$. The latter is equal to the following gate,
\begin{equation}
    \includegraphics[scale=0.44]{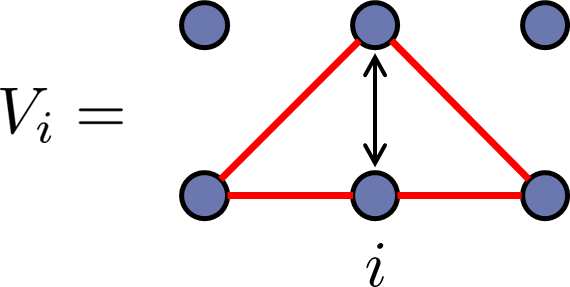}
\end{equation}
where the lower (upper) row of qubits belongs to the $A$ ($B$) system, and the four $CZ$ gates (red lines) act before the SWAP gate (arrow).
This particular product of CZ gates commutes with the $\mathbb{Z}_2\times\mathbb{Z}_2$ symmetry, as do the SWAP gates $S_i$. It is straightforward to check that $V_R = (\prod_{i\in R} S_i)(\prod_{i\in R} V_i)$ is equal to the following operator,
\begin{equation}
    \includegraphics[scale=0.44]{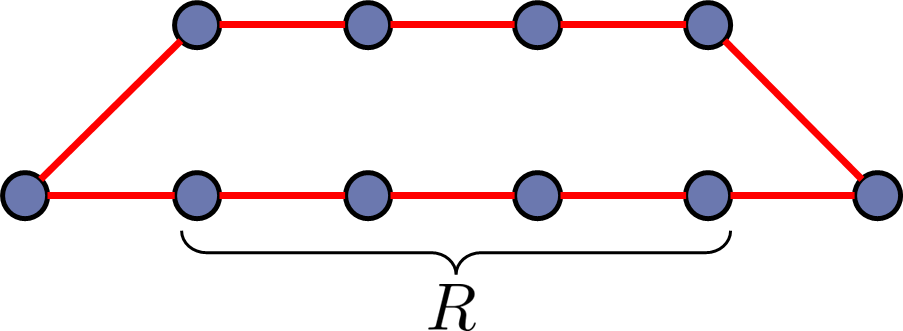}
\end{equation}
which is exactly $Q$ applied to a finite ring. Therefore, this gives an alternative construction for a symmetric \fdqck realizing the 1D cluster state which is similar to, but distinct from, the circuit in Eq.~\ref{eq:1dspt_disent}.

\subsection{2D diagonal shift QCA} \label{sec:qca_ex_2dshift}

Here we show how our construction applies to higher-dimensional QCA by considering the 2D diagonal-shift QCA which translates all sites up and to the right by one site. This example highlights the need for the inductive part of the proof, where we decompose the QCA into products of lower-dimensional QCA until a \fdqck is created. 

We consider implementing the operation on a 2D torus of size $L_x\times L_y$. As described in Sec.~\ref{sec:qca_higherdim}, we first compactify along the $y$ direction and write the QCA in the standard form. This requires $d=2^{L_y}$, $\ell=1$, and $r=2^{2L_y}$. Then we have, 
\begin{equation}
u=\sum_{\vec{i},\vec{j}} |s(\vec{i})\vec{j}\rangle (\langle \vec{i}|\otimes\langle \vec{j}|),
\end{equation}
and,
\begin{equation}
v=\sum_{\vec{i},\vec{j}} (|\vec{i}\rangle \otimes |s(\vec{j})\rangle)\langle \vec{i}\vec{j}|,
\end{equation}
where $\vec{i}=(i_1,\dots i_{L_y})$ and $s(\vec{i}) = (i_{L_y}, i_1,\dots i_{L_y-1})$ is a vertical translation of $\vec{i}$. The spatial inversions $\bar{u}$ and $\bar{v}$ are defined as in the case of the 1D shift QCA. We then calculate,
\begin{equation}
\begin{aligned}
    w &=\sum_{\vec{i},\vec{j}} (|s(\vec{j})\rangle \otimes |s(\vec{i}))\rangle)(\langle \vec{i}|\otimes \langle \vec{j}|) \\
    &=(\mathrm{SHIFT}_y\otimes \mathrm{SHIFT}_y) \overrightarrow{\mathrm{SWAP}}
\end{aligned}
\end{equation}
where $\overrightarrow{\mathrm{SWAP}}$ swaps the two columns indexed by $\vec{i}$ and $\vec{j}$, and $\mathrm{SHIFT}_y$ translates all sites in a column up by one. We observe here that $w$ is composed of lower-dimensional QCA, and is not a FDQC itself. This illustrates the need for the inductive part of the proof. We can express $w$ as an \fdqck as shown in Sec.~\ref{sec:qca_ex_shift}. 

To obtain the complete \fdqck realizing the diagonal shift QCA, we combine $V_R$, which is just a sequence of two products of pairwise SWAPs between sites, with products of $w$ and $\bar{w}$ as described in Eq.~\ref{eq:qcafdqck}.
}

\section{Properties of the operator $w$} \label{app:w}
The following Lemma is used in Sec.~\ref{sec:qca} to show that the operator $w$ is locality-preserving.
\begin{lemma}
    Consider a system of qubits divided into $A$ and $B=\bar{A}$ regions. Let $O$ be an operator acting on the whole system. Using Schmidt decomposition, one can always write $O$ as,
        \begin{align}
            O=\sum_{k=1}^M A_k \otimes B_k,
        \end{align}
        such that $A_k$ and $B_k$ are operators acting on regions $A$ and $B$ respectively and, all $A_k$ operators are linearly independent and all $B_k$ operators are also linearly independent.  Assume $O$ acts trivially on a qubit $j$. Then, if $O$ is written as above, each $A_k$ and each $B_k$ should act trivially on qubit $j$ as well.
    \begin{proof}
    Without loss of generality, we assume $j\in A$, from which it trivially follows that $B_k$ operators act as identity on $j$. Let $P_j$ be an arbitrary operator supported on qubit $j$. Since $O$ acts trivially on $j$, we have $[O,P_j]=0$ and thus,
    \begin{align}
        \sum_{k=1}^M [A_k,P_j]\otimes B_k=0.
    \end{align}
    Let $|\psi \rangle$ be an arbitrary state in the region $A$. If we multiply both sides of the above equality by $|\psi\rangle \langle \psi|\otimes \mathbb{I}_B$ and trace over $A$ we get,
    \begin{align}
        \sum_{k=1}^M \langle \psi |[A_k,P_j] | \psi \rangle B_k=0
    \end{align}
    Since $B_k$s are linearly independent, we find that $\langle \psi |[A_k,P_j] | \psi \rangle=0$ for all $k$. Since $|\psi \rangle$ was arbitrary, it follows that $[A_k,P_j]=0$. Finally, since $P_j$ was arbitrary, we conclude that $A_k$ should act trivially on qubit $j$. 

    \end{proof}
\end{lemma}

In the following Lemma, we show that the unitary $w$ defined in Sec.~\ref{sec:qca} commutes with all global symmetries of $Q$.

\begin{lemma}
 If $[Q,S]=0$ where $S = s^{\otimes N}$ then $[w,S]=0$.
 \begin{proof}
 Recall the Margolus representation of $Q$ in Eq.~\ref{eq:standard} which defines $Q$ in terms of matrices $u$ and $v$. If $[Q,S]=0$, then $S^{-1} Q S = Q$ so $u'=u(s\otimes s)$ and $v'=(s^{-1}\otimes s^{-1})v$ define the same QCA as $u$ and $v$. Then, by Theorem 3.10 of \cite{Cirac2017}, there must exist unitaries $x$ and $y$ such that $u'=(x\otimes y)u$ and $v'=v(y^{-1}\otimes x^{-1})$ (note that Eq.~35b in \cite{Cirac2017} implies the condition on $v$ via Eq.~29b therein). Defining $w'=\bar{v'}u'$, we have, $$w'=\bar{v}(x^{-1}\otimes y^{-1})(x\otimes y)u = \bar{v}u\equiv w,$$ so, $$w'\equiv(s^{-1}\otimes s^{-1})w(s\otimes s)=w,$$ which gives $[w,S]=0$.
 \end{proof}
\end{lemma}
Lastly, in the following Lemma, we argue that if $Q$ is transitionally invariant in the compactified direction, then $w$ is also transitionally invariant in the compactified direction.
\begin{lemma} \label{lemma:trans}
    If $T$ is a translation along the compactified directions, then for any operator $O$ supported on supersites $2i$ and $2i+1$, we have $w T O T^\dagger w^\dagger =T w O w^\dagger T^\dagger$. 
\end{lemma}
\begin{proof}
Let $O$ be an operator supported on supersites $2i$ and $2i+1$. Let $ u O u^{-1} =\sum_k A^k \otimes B^k$, where $A^k$ ($B^k$) are linearly independent operators supported on the $l$-dimensional ($r$-dimensional) Hilbert space that comes out of $u$. We also define $C^k=v A^k v^{-1}$ and $D^k=v B^k v^{-1}$. On the other hand, let $u (TOT^\dagger)u^{-1}=\sum_k \widetilde{A}^k \otimes \widetilde{B}^k$, where $\widetilde{A}^k$ ($\widetilde{B}^k$) are linearly independent operators supported on the $l$-dimensional ($r$-dimensional) Hilbert spaces. We also define $\widetilde{C}^k=v \widetilde{A}_k v^{-1}$ and $\widetilde{D}^k=v \widetilde{B}^k v^{-1}$. Note that $Q O Q^{-1}=\sum_k C^k \otimes D^k$, where $C^k$ operators act on supersites $2i-1$ and $2i$ and $D^k$ acts on supersites $2i+1$ and $2i+2$. On the other hand, we have $Q (TOT^\dagger) Q^{-1}= \sum_k \widetilde{C}^k\otimes \widetilde{D}^k$. Since $Q$ is translationally invariant, we have
\begin{align}\label{eq:translationCD}
    \sum_k \widetilde{C}^k \otimes \widetilde{D}^k=\sum_k T C^k T^\dagger \otimes T D^k T^\dagger.
\end{align}
Now, when we act on $TOT^\dagger$ with $w$ we get,
\begin{align}
 w (TOT^\dagger) w^\dagger=\sum_k \overline{\widetilde{C}^k}~\overline{\widetilde{D}^k}=\text{SWAP} \Big[\sum_k \widetilde{C}^k\widetilde{D}^k \Big] \text{SWAP},
\end{align}
where $\text{SWAP}$ exchanges the two super sites on which $w$ acts. However, it follows from  Eq.~\ref{eq:translationCD} 
 that $\sum_k\widetilde{C}^k\widetilde{D}^k=\sum_k T C^k D^k T^\dagger$ which is easy to see graphically,
 \begin{align}
     &\includegraphics[width=0.95\columnwidth]{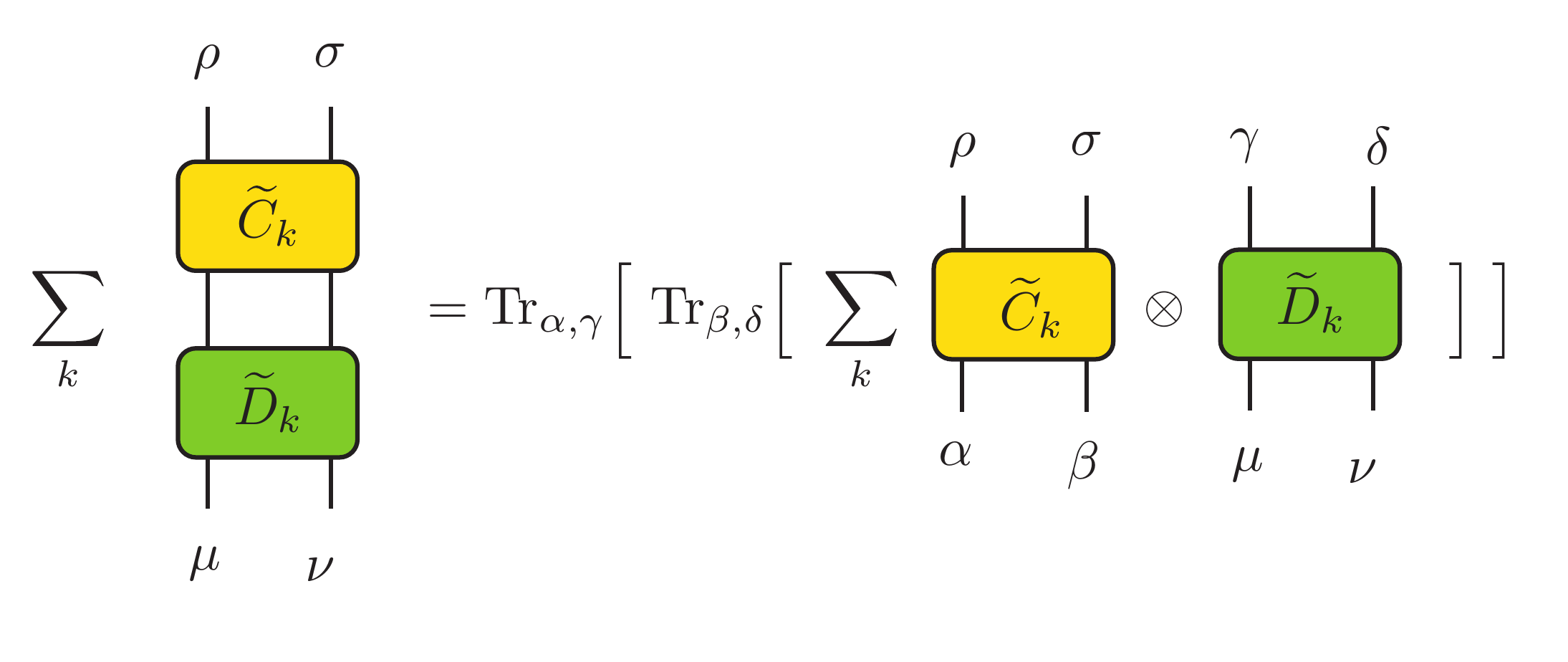}\nonumber\\
     &\includegraphics[width=0.95\columnwidth]{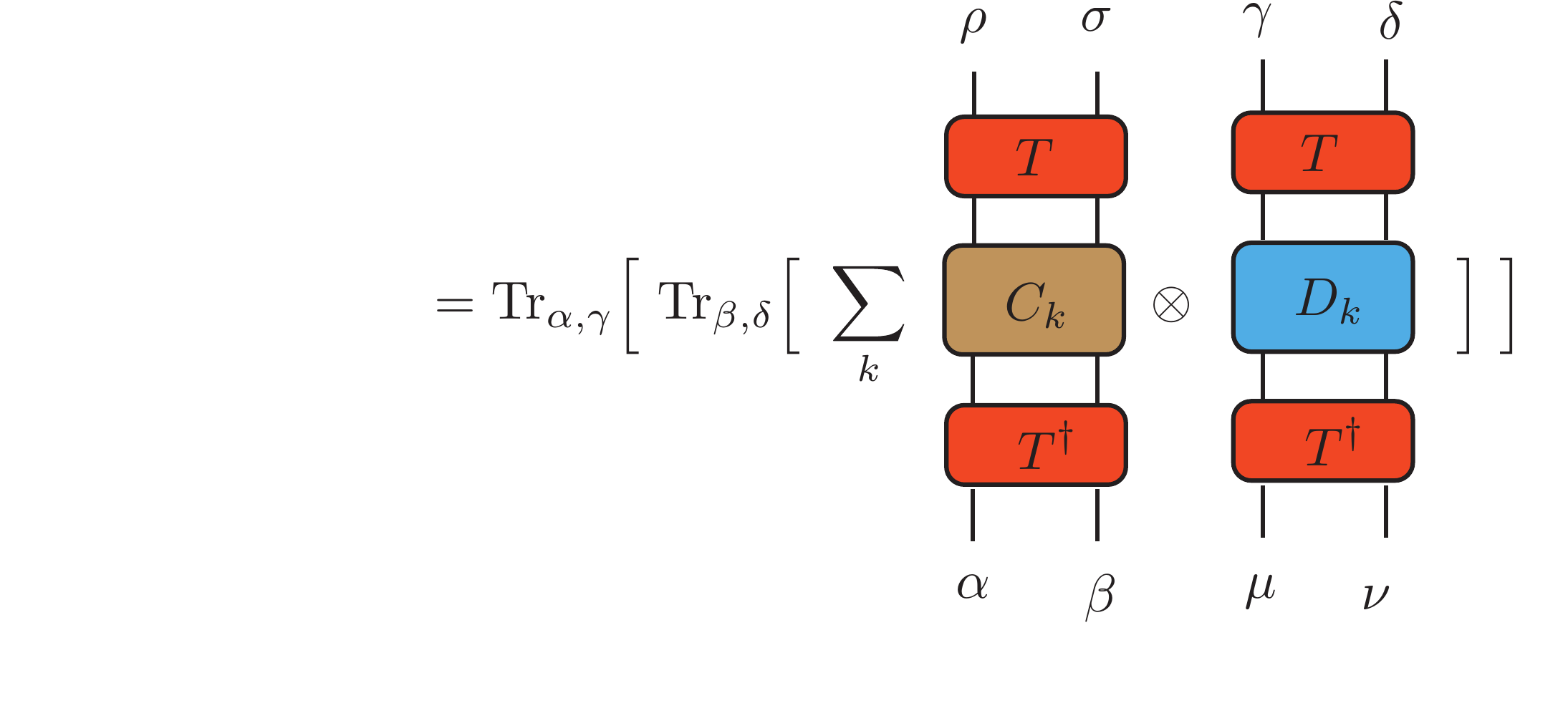}\nonumber\\
     &\includegraphics[width=0.95\columnwidth]{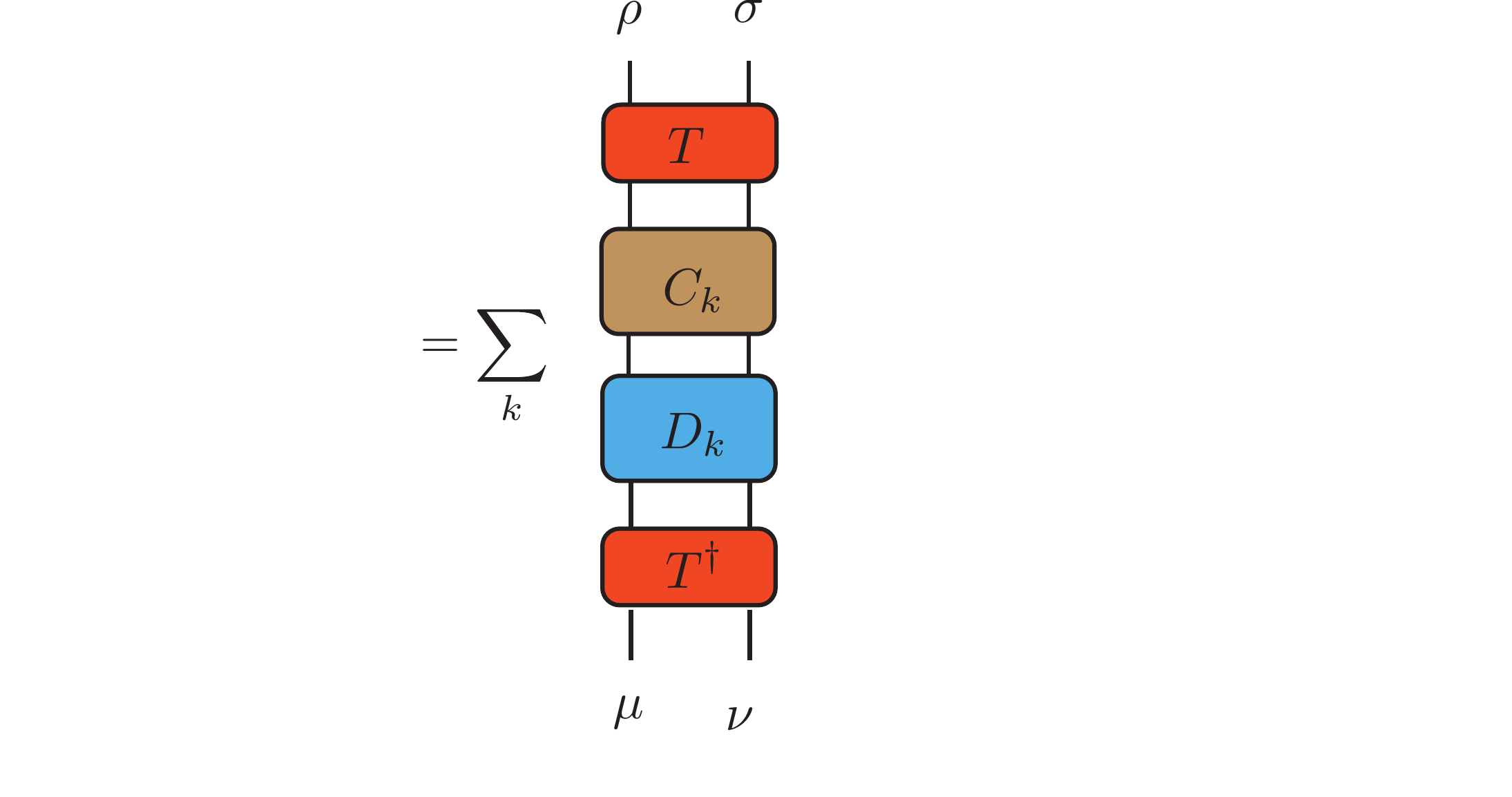}
 \end{align}
 where $\text{Tr}_{\alpha,\beta}$ means contracting the $\alpha$ and $\beta$ indices. In the first line, we have used the linearity of the trace to move the sum over $k$ inside the trace. To go to the second line, we have used Eq.~\ref{eq:translationCD}, and the third line follows by moving the sum outside the traces and contracting the relevant indices. Lastly, since the $\text{SWAP}$ between supersites commutes with the translation $T$ along the compactified directions, we find that,
 \begin{align}
     w (TOT^\dagger) w^\dagger&=\text{SWAP}~T \Big[\sum_k C^k D^k \Big]T~\text{SWAP}\nonumber\\
     &=T\Big[\sum_k \overline{C^k}~\overline{D^k}\Big]T \nonumber\\
     &=T~wOw^\dagger ~ T^\dagger.
 \end{align}
\end{proof}

\section{Light cone argument for $k$-local non-triviality}\label{app:QECC_klocal_nontriviality}
In this section, we briefly review the proof for $k$-local non-triviality of code states of quantum error correcting codes. The argument is basically the same as the light cone argument presented in Ref.~\cite{Bravyi2006}, which was used to show that topological states on manifolds of non-zero genus can not be prepared by constant depth local unitaries, but with the slight modification of replacing locality with $k$-locality. 

Let $\ket{\psi_1}$ and $\ket{\psi_2}$ be two orthogonal code states of a $N$ qubit quantum code. Furthermore, assume the corresponding quantum error correcting code has distance $d$, so for any operator which acts on less than $d$ qubits, we have,
\begin{align}\label{eq:ecc}
    \langle \psi_1 | O |\psi_1 \rangle =\langle \psi_2 | O |\psi_2 \rangle.
\end{align}
Let $U$ be a $k$-local circuit that prepares $\ket{\psi_1}$ in depth $D$, starting from the trivial state $\ket{0}^{\otimes N}$,
\begin{align}\label{eq:qeccU}
    \ket{\psi_1}=U \ket{0}^{\otimes N}
\end{align}
In the following, we show that $D\ge \log_k(d)$. Assume it is not true, meaning that $D< \log_k(d)$. Let $\pi_j=\ketbra{0}_j$ denote the projection operator that projects the $j$'th qubit into the $\ket{0}$ state. Note that $\pi_j$ has only support on qubit $j$. Since $U$ is a $k$-local circuit of depth $D$, the operator $U \pi_j U^\dagger$ can have non-trivial support on at most $k^D$ qubits, which is less than $d$ (because of the assumption $D<\log_k(d)$). Therefore, due to Eq.~\ref{eq:ecc} we have,
\begin{align}
    \expval{U \pi_j U^\dagger}{\psi_2}=\expval{U \pi_j U^\dagger}{\psi_1}=1,
\end{align}
where we have used Eq.~\ref{eq:qeccU} in the last step.  Since $j$ was arbitrary, we should have $U^\dagger \ket{\psi_2}=e^{i\theta}\ket{0}^{\otimes N}$, for some phase  $\theta$, or equivalently $\ket{\psi_2}=e^{i\theta}\ket{\psi_1}$. But this contradicts the assumption that $\ket{\psi_1}$ and $\ket{\psi_2}$ are orthogonal code states, hence $D \ge \log_k(d)$.

The above argument shows that the code states of quantum error correcting codes with a non-zero number of logical qubits, whose distance goes to infinity in the thermodynamic limit, are $k$-local non-trivial, meaning that they cannot be prepared by a $k$-local constant depth unitary circuit. This includes for example the ground states of the toric code on a torus. However, the argument above says nothing about the complexity of preparing the ground state of the topological Hamiltonians such as the toric code on a sphere that has genus zero, because in this case, the ground state is unique (so there is no other orthogonal ground state $\ket{\psi_2}$). Similarly, the light cone argument does not work for SPT states on closed manifolds since the ground state of an SPT Hamiltonian on a closed manifold is unique. While a more involved argument\cite{Aharonov2018} shows that topological states on zero genus surfaces are still $k$-local non-trivial, our result shows that SPT states in contrast are all $k$-local trivial. It is worth noting that although SPT phases on manifolds with open boundary conditions have degenerate ground states, the light cone argument is still inapplicable. This is because, while Eq.~\ref{eq:ecc} holds for symmetric local operators, it can be violated by symmetric but $k$-local  operators. 

{

\section{More details on the numerical study of SPT order in monitored random circuits}\label{app:spt_numerics}
\begin{figure}
    \centering
    \subfigure[]{{\includegraphics[width=0.49\linewidth]{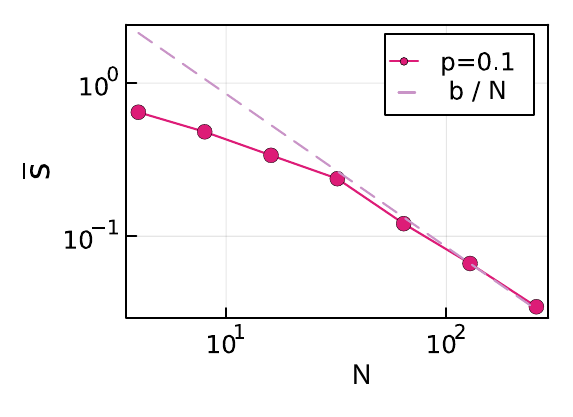}}}
    \subfigure[]{{\includegraphics[width=0.49\linewidth]{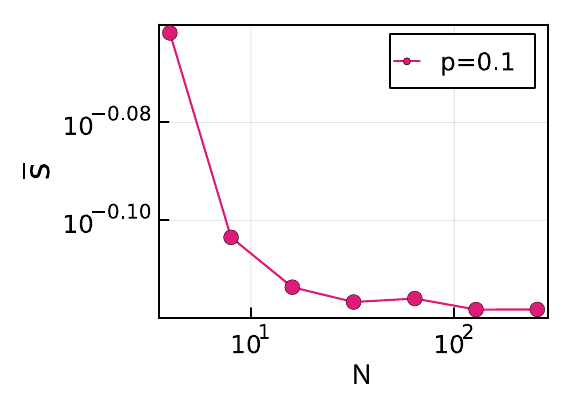}}}
    \subfigure[]{{\includegraphics[width=0.49\linewidth]{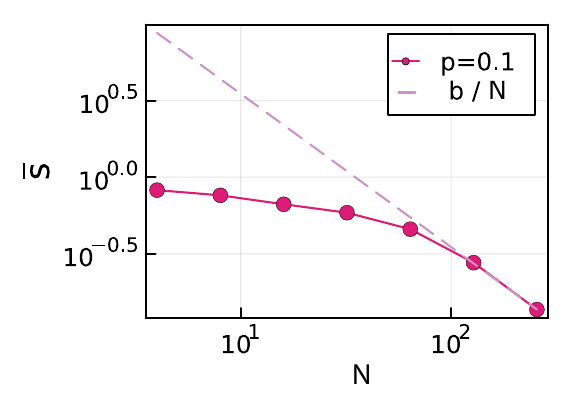}}}
    \subfigure[]{{\includegraphics[width=0.49\linewidth]{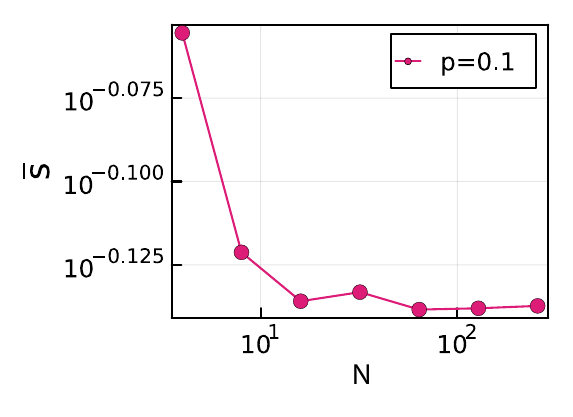}}}
    \caption{
    Average string order parameter $\bar{s}$ versus $N$ at fixed $p=0.1$ in the late-time states of monitored random circuits described in Section \ref{sec:randomcirc} where the unitary gates are chosen randomly from (a) local two-qubit Clifford unitaries, (b) local two-qubit $\mathbb{Z}_2 \times \mathbb{Z}_2$ symmetric Clifford unitaries, (c) 2-local $\mathbb{Z}_2 \times \mathbb{Z}_2$ symmetric Clifford unitaries, and (d) 2-local $\mathbb{Z}_2^N$ symmetric Clifford unitaries.
    }
    \label{fig:sbar_plots_fixed_p}
\end{figure}

In this section, we provide additional details about the transitions that were discussed in Section \ref{sec:randomcirc}. The order parameter that we used to probe the late-time states was given as,
\begin{align}\label{eq:sgp}
    s=\frac{2}{N(N-1)}\sum_{a<b}S(a,b)^2,
\end{align}
with $S(a,b)$ denoting the string order parameter defined in Eq.~\ref{eq:sop}. A state in the SPT phase is characterized by a finite non-zero value of $S(a,b)$ for sufficiently far apart $a$ and $b$. In this case, the sum in Eq.~\ref{eq:sgp} is dominated by sites $a$ and $b$ that are $O(N)$ far apart, and since there are $O(N^2)$ number of such $(a,b)$ pairs, one expects that the parameter $s$ goes to some finite non-zero value in the thermodynamic limit. On the other hand, for a trivial state $S(a,b)$ goes to zero exponentially fast, so only the local terms with $b-a$ smaller than correlation length contributes to the sum in Eq.~\ref{eq:sgp}, and since there are $O(N)$ of such terms, one expects $s$ to drop as $1/N$ for large $N$ and goes to $0$ in the thermodynamic limit. This scaling can be seen in our setup as expected. Fig.~\ref{fig:sbar_plots_fixed_p} shows the order parameter $s$ as a function of $N$ for fixed $p=0.1$, with both axes scaling logarithmically. Note that $s$ vanishes as $1/N$ for large $N$ in panel (a) and (c), which correspond to local (but not necessarily symmetric) $2$-qubit unitaries and $2$-local symmetric unitaries respectively, while it saturates to a finite non-zero value in panel (b) and (d) which correspond to local symmetric $2$-qubit unitaries and $2$-local ``locally" symmetric unitaries (see Section \ref{sec:randomcirc}) respectively.

\begin{figure}[b]
    \centering
    \subfigure[]{{\includegraphics[width=0.49\linewidth]{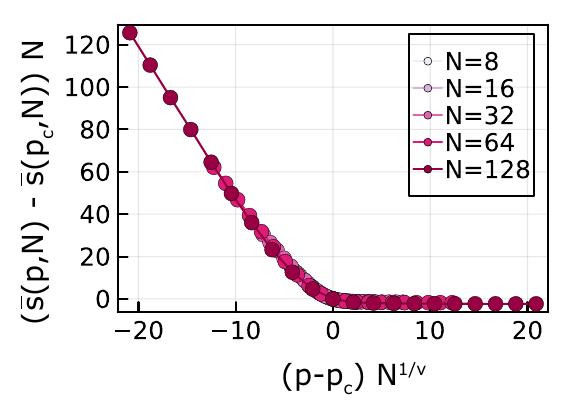}}}
    \subfigure[]{{\includegraphics[width=0.49\linewidth]{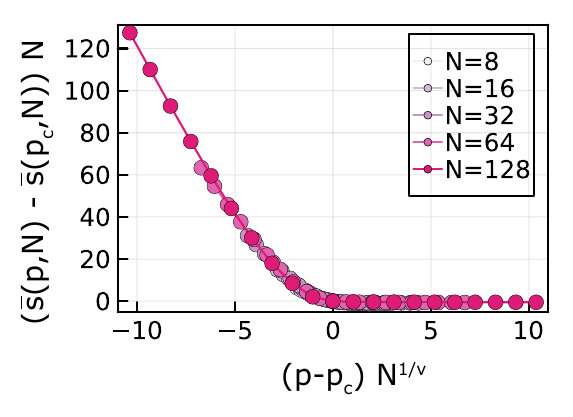}}}
    \caption{
    Data collapse of data points shown in Fig.~\ref{fig:sbar_plots} for (a) when gates are chosen to be local symmetric Cliffords, with $p_c=0.50(5)$ and $\nu=1.3(2)$ and (b) when gates are chosen to be $2$-local locally symmetric Cliffords with $p_c=0.50(5)$ and $\nu=1.6(1)$. 
    }
    \label{fig:s_datacollapse}
\end{figure}
Moreover, for monitored circuits where the SPT order survives up to nonzero values of $p$, i.e. for local symmetric unitaries and $2$-local unitaries that are locally symmetric (panels (b) and (d) in Fig.~\ref{fig:sbar_plots_fixed_p}), one can study the phase transition at $p_c$ by using data collapse for finite size systems. Following Ref.~\cite{sang2021measurement} we assume the following scaling form near the critical point,
\begin{align}
    s(p,N)-s(p_c,N)=N^{-1}~F[(p-p_c)~N^{1/\nu}],
\end{align}
where $p_c$ is the critical value for applying unitary gates and $\nu$ is the correlation length critical exponent, and we can search for values of $p_c$ and $\nu$ that result in the best data collapse. Fig.~\ref{fig:s_datacollapse} shows the best collapse using the data points shown in Fig.~\ref{fig:sbar_plots}, finding $p_c=0.50(5)$ for both local symmetric gates as well as $2$-local locally symmetric gates. On the other hand, the correlation length critical exponent $\nu$ for local symmetric gates and $2$-local locally symmetric gates are $\nu=1.3(2)$ and $\nu=1.6(1)$. In principle one needs to get more data near the critical point and use them for data collapse to get better estimates of $p_c$ and $\nu$.
}

\end{document}